\documentclass[lettersize,journal]{IEEEtran}

\usepackage{url}
 
\usepackage{amsmath}
\usepackage{amsthm}
\usepackage{float} 

\usepackage[numbers,sort&compress]{natbib}

\usepackage{graphicx}
\usepackage{amssymb,bm}
\usepackage{algorithm}
 
\usepackage[noend]{algcompatible}

\usepackage{tablefootnote}
\usepackage{multirow}
\usepackage{multicol}
\usepackage{booktabs}
\usepackage{enumitem}
\usepackage{marvosym}

\usepackage{courier}

\usepackage{pifont}
\usepackage{tcolorbox}
\newcommand{\cmark}{\ding{51}}%
\newcommand{\xmark}{\ding{55}}%

\algdef{SE}[FUNCTION]{Function}{EndFunction}{}{\algorithmicend\ }
\algtext*{Function}
\algtext*{EndFunction}

\newcommand{\RecastCheck}{\textsc{RcLock}}
\newcommand{\RecastVal}{\textsc{RcStore}}

\newcommand{\If}{\IF}
\newcommand{\Else}{\ELSE}
\newcommand{\EndIf}{\ENDIF}

\algrenewcommand\algorithmicwhile{\textbf{upon}}
\newcommand{\Upon}{\WHILE}
\newcommand{\EndUpon}{\ENDWHILE}

\algrenewcommand\algorithmicloop{\textbf{wait}}
\newcommand{\Wait}{\LOOP}
\newcommand{\EndWait}{\ENDLOOP}

\newtheorem{definition}{Definition}
\newtheorem{theorem}{Theorem}
\newtheorem{lemma}{Lemma}

\newcommand{\Sign}{\mathsf{Sign}}

\newcommand{\plusplus}{\mathsf{Dumbo}\textrm{-}\mathsf{NG}\texttt{++}}
\newcommand{\dumbo}{\mathsf{Dumbo}}
\newcommand{\tusk}{\mathsf{Tusk}}
\newcommand{\dumbong}{\mathsf{Dumbo}\textrm{-}\mathsf{NG}}
\newcommand{\finng}{\mathsf{FIN}\textrm{-}\mathsf{NG}}
\newcommand{\jumbo}{\mathsf{JUMBO}}

\newcommand{\mvba}{\mathsf{MVBA}}
\newcommand{\rbc}{\mathsf{BRBC}}
\newcommand{\raba}{\mathsf{RABA}}

\newcommand{\node}{\mathcal{P}}
\newcommand{\hash}{\mathcal{H}}
\newcommand{\adv}{\mathcal{A}}

\newcommand{\bigO}{\mathcal{O}}

\newcommand{\PD}{\mathsf{PD}}
\newcommand{\RC}{\mathsf{RC}}
\newcommand{\store}{{store}}
\newcommand{\lock}{{lock}}
\newcommand{\DID}{\mathsf{ID}}
\newcommand{\ValidateLock}{\mathsf{ValidateLock}}
\newcommand{\PreDisp}{\textsc{Store}}

\newcommand{\Store}{\textsc{Stored}}

\newcommand{\Agg}{\mathsf{Agg}}

\newcommand{\Verify}{\mathsf{SigVerify}}

\newcommand{\height}{\mathsf{current}}
\newcommand{\ordered}{\mathsf{ordered}}
\newcommand{\newordered}{\mathsf{new\textrm{-}ordered}}

\newcommand{\tx}{\mathsf{tx}}

\newcommand{\ignore}[1]{}
\newcommand{\yuan}[1]{\textcolor{black}{#1}}

\begin{document}
	
   \title{$\jumbo$: Fully Asynchronous BFT Consensus Made Truly Scalable}

\author{Hao Cheng, Yuan Lu, Zhenliang Lu, Qiang Tang, Yuxuan Zhang, Zhenfeng Zhang
\thanks{
\IEEEcompsocthanksitem Hao Cheng, Yuan Lu, and Zhenfeng Zhang are with Institute of Software  Chinese Academy of Sciences, Beijing, 100190, China. Email: {\{chenghao2020, luyuan, zhenfeng\}@iscas.ac.cn}. 
\IEEEcompsocthanksitem Zhenliang Lu and Qiang Tang are with  The University of Sydney, NSW 2006, Australia. Email: {\{zhenliang.lu, qiang.tang\}@sydney.edu.au}. 
\IEEEcompsocthanksitem Yuxuan Zhang is with University of Pittsburgh,  Pittsburgh, 15260, PA, US. Email: {yuz276@pitt.edu}. 
\IEEEcompsocthanksitem Hao Cheng, Yuan Lu, and Yuxuan Zhang made equal contribution, and part of the work was done while Yuxuan visited Institute of Software CAS.
}

}


\maketitle

	\begin{abstract}
	Recent progresses in asynchronous Byzantine fault-tolerant (BFT) consensus, e.g. $\dumbong$ (CCS' 22) and $\tusk$ (EuroSys' 22),   show  promising performance through decoupling transaction dissemination and block agreement. However, when executed with a larger number $n$ of nodes, like several hundreds, they would  suffer from significant  degradation in performance.
	Their dominating scalability bottleneck is  the huge authenticator complexity: each node has to multicast $\bigO(n)$  quorum certificates (QCs) and subsequently verify  them  for each block.
	
	This paper systematically investigates and resolves the above scalability issue. We first propose a signature-free asynchronous BFT consensus  FIN-NG that adapts a recent signature-free asynchronous common subset protocol FIN (CCS' 23) into the   state-of-the-art framework of concurrent  broadcast and agreement. The liveness of FIN-NG relies on our non-trivial redesign of FIN's multi-valued validated Byzantine agreement    towards achieving optimal quality.
	FIN-NG greatly improves the performance of FIN and already outperforms $\dumbong$  in most deployment settings. 
	To further overcome the scalability limit of FIN-NG due to $\bigO(n^3)$ messages, we propose $\jumbo$, a scalable instantiation of $\dumbong$, with only $\bigO(n^2)$ complexities for both authenticators and messages. We use various aggregation and dispersal techniques for QCs to significantly reduce the authenticator complexity of original $\dumbong$ implementations  by up to $\bigO(n^2)$ orders. We also propose a ``fairness'' patch for $\jumbo$, thus preventing a flooding adversary from controlling an overwhelming portion of transactions in its output.
	

	
	Finally, we  implement  our designs in Golang and experimentally demonstrated their enhanced scalability with hundreds of Amazon's  AWS  instances. $\jumbo$ and FIN-NG significantly outperform  the  state-of-the-art in (nearly) all deployment settings. 
	Especially,	when $n\ge$196, $\jumbo$ can attain a throughput that is more than 4x that of FIN and $\dumbong$. 
	
\end{abstract}



\section{Introduction}
\label{sec:intro}

Randomized  {\em asynchronous} Byzantine-fault tolerant (BFT) consensus protocols \cite{ben1983another,rabin1983randomized,MMR15,canetti1993fast,cachin2002secure,moniz2008ritas,patra2009simple,abraham2008almost,cachin00,ben2003resilient,cachin2001secure,miller2016honey,beat,guo2020dumbo,abraham2018validated,tusk,lu2020dumbo,gao2022dumbo,abraham2021reaching,duan2023practical,yang2021dispersedledger} can overcome   FLP ``impossibility'' \cite{FLP85} to    ensure both {\em safety} (i.e. all honest nodes agree the same ledger of transactions) and {\em liveness} (i.e. a client can expect her valid transactions eventually output) despite a fully asynchronous network that can arbitrarily delay  message delivery. In contrast, their synchronous or partial synchronous counterparts, such as PoW \cite{bitcoin} and PBFT \cite{pbft}, might suffer from inherent violation of security (e.g., PoW might output disagreed decisions \cite{saad2021revisiting} and PBFT can completely grind to a halt \cite{miller2016honey}) if unluckily encountering an asynchronous adversary \cite{FLP85}. Therefore,  the fully asynchronous BFT consensus protocols become highly desired for their high security assurance, making them the most robust candidates for critical blockchain infrastructures deployed over the unstable and occasionally adversarial global Internet.

\subsection{Scalability Obstacles of  the State of the Art}


Interests in asynchronous BFT consensus protocols have  recently resurged  \cite{miller2016honey,beat,guo2020dumbo,abraham2018validated,lu2020dumbo,duan2023practical,yang2021dispersedledger,tusk,gao2022dumbo,abraham2021reaching} for their higher security assurance against powerful network adversaries.  However, there is limited evidence supporting their seamless scalability to efficiently handle larger-scale networks with several hundred nodes. For instance, the state-of-the-art asynchronous BFT consensus protocols like  $\tusk$ \cite{tusk}\footnote{https://github.com/asonnino/narwhal} and  $\dumbong$ 
 \cite{gao2022dumbo} were  demonstrated with only 50 and 64 nodes, respectively. As Figure \ref{fig:scale} depicts, we evaluate these state-of-the-art designs    in a LAN setting,
revealing that they suffer from a significant performance decline as the number of participating nodes increases, despite their appealing performance at smaller scales like 16 or 64 nodes. \footnote{Note that a direct comparison between the throughputs of $\tusk$ and $\dumbong$  is meaningless due to their difference on implementing storage, and we shall focus on the trend of throughput and latency change with scales.} Notably, while the number of    nodes $n$ increases from 16 to 256, their peak throughput   and latency are worse by   3 and 2 orders of magnitude, respectively.

\begin{figure}[H] 
	\vspace{-0.25cm}
	\begin{flushleft}
		\includegraphics[width=8.5cm]{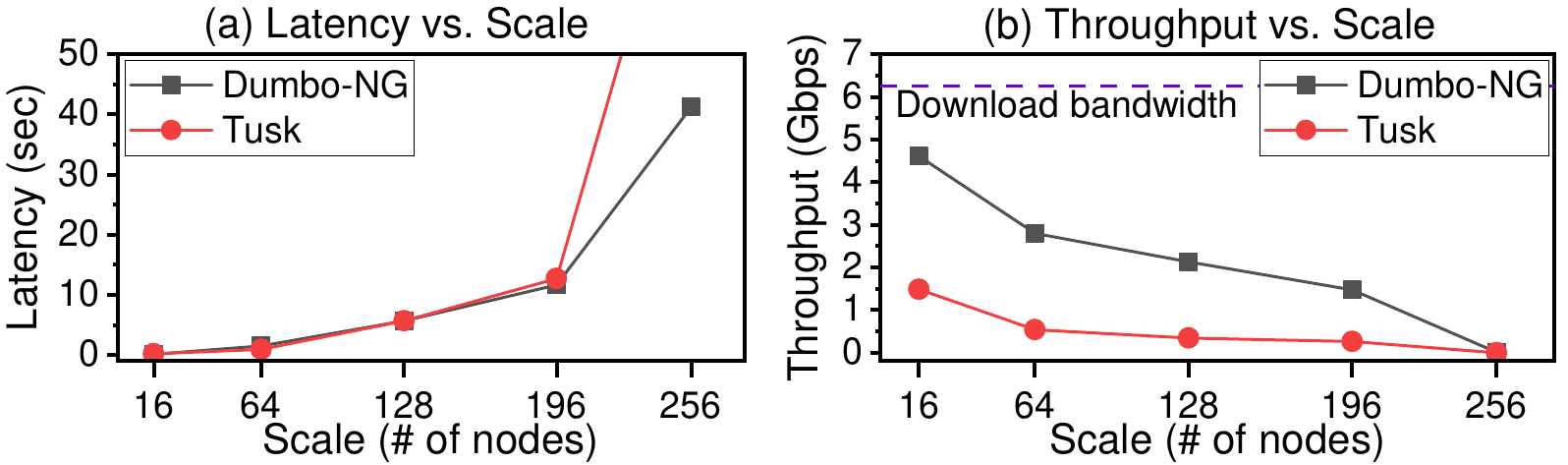}\\
		\vspace{-0.3cm}
		\caption{$\tusk$ \cite{tusk} and $\dumbong$ \cite{gao2022dumbo}   for 16-256 nodes in  LAN. Each node is an   EC2 c6a.2xlarge instance with  12.5 Gbps bandwidth (upload+download).}
		\label{fig:scale}
	\end{flushleft}
	\vspace{-0.35cm}
\end{figure}

The inferior performance in larger-scale networks significantly hampers the widespread adoption of asynchronous BFT consensus in real-world blockchain systems. In particular, a larger number of participating nodes is a sine qua non to diversify trust for  mission-critical applications such as cryptocurrency and decentralized finance, thus  necessitating our continued efforts to boost the performance of asynchronous   consensus in larger-scale networks.

\smallskip
\noindent
{\bf Scalability bottlenecks of   signature-involved prior art}.  We begin with a succinct examination of $\dumbong$ and $\tusk$ ---   the most practical asynchronous BFT consensus protocols in relatively small networks such as several dozens of nodes, in order to unveil the major obstacles hindering their scalability.
Both  $\tusk$ and $\dumbong$ can execute a process continuously disseminating transactions concurrently to another process deciding a unitary order of  transactions that have been   disseminated, thus maximizing   bandwidth utilization and achieving appealing performance (for several dozens of nodes).

 \ignore{
	In $\dumbong$, the transaction dissemination is ``provable'': every node executes as a broadcast leader and can generate quorum certificates (QCs) consisting of sufficient digital signatures (e.g., $n-f$) to attest how many transactions it has successfully broadcasted; Concurrently, every node records the lastest broadcast QC sent from each participating node, and uses these $n$ latest QCs as input to some multi-valued validated Byzantine agreement ($\mvba$) module (which is guaranteed to output $n$ valid broadcast QCs), such that all nodes can make a unitary decision according to the $\mvba$ output on how to cut the disseminated transactions into a linearized sequence of blocks. In $\tusk$, the transaction dissemination phase is not ``provable'': a node diffuses its input transaction and just waits for sufficient naive acknowledgments instead of digital signatures. Meanwhile, $\tusk$ leverages an asynchronous directed acyclic graph (DAG) protocol to order those diffused transactions in a concurrent manner. Here each DAG ``vertex'' actually represents a consistent broadcast protocol that allows a designated leader to generate and send a QC to convince all participating nodes that it sends to them the same transactions' digests. Moreover, every DAG ``vertex'' has to link to sufficient  precedent ``vertices'' created by distinct participating nodes (e.g. $n-f$). 
	That means, when each node executes a consistent broadcast in $\tusk$'s DAG, it needs to broadcast at least $n-f$ QCs (playing the role of $n-f$ links in DAG).
}

However, both protocols heavily rely on signature based quorum certificates (QCs) to implement the performance-optimized structure of concurrent broadcast and agreement: each node has to broadcast $\bigO(n)$ QCs on different messages per decision,   indicating overall $\bigO(n^3)$ QCs per output block.
Worse still, 
$\dumbong$ and $\tusk$, like many  earlier BFT studies \cite{yin2018hotstuff-full,guo2022speeding,bft-smart},  implement QCs by  concatenating $2f+1$ non-aggregate signatures (where $f$ is the number of maximal malicious nodes).\footnote{It might be a little bit surprising, but we find that the official implementations of many popular BFT protocols (like HotStuff, Speeding Dumbo, Dumbo-NG, Narwhal and Tusk) use ECDSA/EdDSA for implementing QC.} Although they can theoretically use  aggregatable signatures   like   threshold Boneh-Lynn-Shacham (BLS) signature as alternative, they explicitly avoid \yuan{the aggregatable threshold signatures} in practice, because threshold \yuan{signatures} might entail concretely   expensive computation costs of cryptographic pairings and Lagrange interpolation in the exponent (as elaborated by a very recent thorough benchmarking \cite{li2023performance}).



Then if we implement QCs by ECDSA/EdDSA with standard 128-bit security \yuan{as suggested by earlier practices \cite{yin2018hotstuff-full,guo2022speeding,bft-smart,li2023performance}}, when $n$=$200$ nodes, 
  \yuan{each node needs to send  at least  343 MB signatures} per    decision;\footnote{\yuan{Each QC contains $n$-$f$=134 ECDSA/EdDSA signatures, and each   signature has 64 bytes (if compression is turned off for computation efficiency). So sending $n$ QC to $n$ nodes would cost $(n-f)*64*n*n$=343 MB.}}
even worse, each node   receives a minimum of $n$ distinct QCs that have to verify, 
indicating a CPU latency  of    second(s).\footnote{\yuan{Verifying  an ECDSA/EdDSA   costs $\sim$40 $\mu$s at EC2 c6a.2xlarge instance.}}
As Figure \ref{fig:ng-break} depicts, we further test $\dumbong$ in LAN to experimentally \yuan{demonstrate} the high cost of QCs in practice.
It reveals: 
(i)  when $n$ increases to 196 from 16, the ratio of   signature verifications in  CPU time increases from $\sim$16.7\% to $\sim$66.6\%,\footnote{Breakdown of CPU time was measured using Golang's profiling tool pprof.}
  (ii) the communication   of QCs jumps from $\sim$0.8\% of total communication to $\sim$40\%. Apparently, the overhead   of signatures (which can be characterized by a fine-grained metric called authenticator complexity \cite{yin2018hotstuff-full,avarikioti2020fnf}) now plays a key role in the inferior performance    in the large-scale   settings with hundreds of nodes.

\begin{figure}[H] 
	\vspace{-0.15cm}
	\begin{center}
		\includegraphics[width=7.8cm]{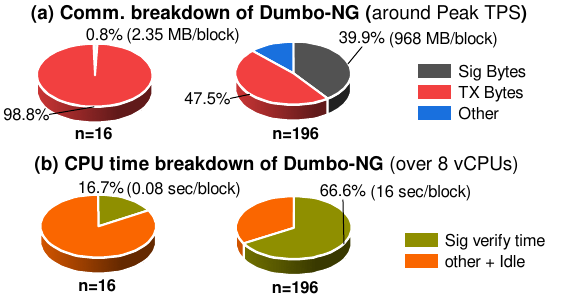}\\
		\vspace{-0.15cm}
		\caption{Cost breakdown       of Dumbo-NG in LAN   consisting of 16-196 nodes.}
		\label{fig:ng-break}
	\end{center}
	\vspace{-0.3cm}
\end{figure}

\noindent
{\bf Practical issues of existing (nearly) signature-free designs}. Since the major scalability bottleneck of $\dumbong$/$\tusk$ stems from their tremendous usage of signature-based QCs, one might wonder whether some existing ``signature-free'' asynchronous BFT protocols could better extend to work in large-scale networks. 
In particular, some recent studies like PACE \cite{zhang2022pace} and FIN \cite{duan2023practical}  attempt to reduce authenticators as many as possible, resulting in some  (nearly) signature-free asynchronous BFT protocols \footnote{We  follow the conventional abuse of term ``signature-free'', though these  protocols \cite{zhang2022pace,duan2023practical}     might require  public key assumptions to practically generate distributed common randomness and/or leverage cryptographic hash functions for efficiency. Both hidden assumptions theoretically imply signatures.} that can get rid of public key operations except for   distributed   randomness. 

In short, 
PACE and FIN directly execute a special variant of asynchronous Byzantine agreement called asynchronous common subset (ACS) \cite{benor}, which can solicit $n-f$ nodes'   proposals as a block of consensus output.
However, these  signature-free  approaches have their own efficiency issues that could be even more serious than $\dumbong$ and $\tusk$  \cite{gao2022dumbo}:
(i) {\em severe transaction censorship}--- the adversary can constantly prevent  $f$ honest nodes from contributing into consensus results, causing  an order of  $\bigO(n)$  communication blowup  because   censored transactions have to be redundantly broadcasted by majority honest nodes;\footnote{HoneyBadger BFT (HBBFT)  uses threshold encryption to wrap ACS against censorship threat  but this introduces   considerable computational cost.}
(ii) {\em low bandwidth utilization}--- they sequentially execute transaction broadcast and agreement, so  during a long period of agreement phase,   bandwidth is ``wasted'' to contribute nothing into throughput; 
(iii) {\em costly erasure-code}--- they use erasure-code based reliable broadcasts \cite{cachin05}, where every node   performs $\bigO(n)$ erasure-code decodings, causing     high CPU cost (cf. Figure \ref{fig:fin-break}).

\begin{figure}[H] 
	\vspace{-0.15cm}
	\begin{center}
		\includegraphics[width=7.8cm]{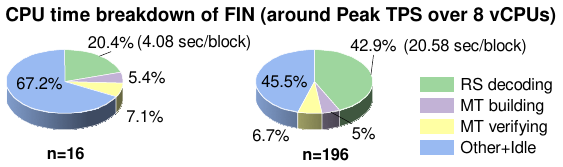}\\
		\vspace{-0.15cm}
		\caption{CPU time breakdown of FIN for 16-196 nodes in LAN (when system load is close to peak throughput). Here ``MT'' represents Merkle tree \cite{merkle1987digital}; ``RS'' represents Reed-Solomon code \cite{reed1960polynomial} (instantiated by a fast implementation \cite{rs-lib} of a systematic Cauchy Reed-Solomon variant \cite{blomer1995xor,plank2005optimizing}).}
		\label{fig:fin-break}
	\end{center}
	\vspace{-0.3cm}
\end{figure}


	\vspace{-0.2cm}

\medskip
Bearing the substantial scalability challenge in the state-of-the-art of asynchronous BFT, the next question remains open:
 	\vspace{-0.6cm}
{
\begin{center}
	\emph{Can we   extend the appealing performance of the \\existing   performant asynchronous BFT consensus protocols \\to larger networks   with several hundreds of nodes?}
\end{center}
}
 	\vspace{-0.45cm}

\begin{table*}[!t]
	\caption{Comparison between asynchronous BFT consensus (state-machine replication) protocols.}\label{tab:compare}
	\vspace{-0.2cm}
	\centering
	\resizebox{2.1\columnwidth}{!}{	
		\renewcommand{\arraystretch}{0.6}	
		\begin{tabular}{lccccclcc}
			
			\toprule
			
			\multirow{2}{*}{} & \multicolumn{5}{c}{\bf Per-Decision Complexities (summed over all nodes)}                                                                                                                                                                    & \multirow{3}{*}{\begin{tabular}[c]{@{}l@{}} \\ Technique to realize  {\bf Liveness}  \\in  the worst case  against tx censorship\end{tabular}} & \multirow{3}{*}{\begin{tabular}[c]{@{}c@{}}\\ {\bf Fairness}\\ (Quality)\end{tabular}} & \multirow{3}{*}{\begin{tabular}[c]{@{}c@{}}  \\ \bf Concurrent\\ \bf broadcast \& \\  \bf agreement\end{tabular}}    \\ \cline{2-6}\rule{0pt}{18pt}
			& \begin{tabular}[c]{@{}c@{}}\bf Amortized \\ \bf Comm. \\ (worst-case)\end{tabular} & \bf Message  & \begin{tabular}[c]{@{}l@{}}\bf Authenticator  \\ (except dist.\\ common coin)\end{tabular} & \bf Round  & \begin{tabular}[c]{@{}l@{}}Number of \\normal-path\\  {\bf EC decoding}\end{tabular}     &                                                                                                                                             &                                                                               &                                                                                                   \\ \hline\rule{0pt}{10pt}
			BKR94~\cite{benor}             & $O(n^3 |\tx|)$                                                                                                       & $O(n^3)$ & -                                                                            & $O(\log n)$ & - & duplicated tx buffers $\ddagger$                                                                                                                   & \cmark                                                                             & \xmark               \\\rule{0pt}{5pt}
			HBBFT~\cite{miller2016honey,beat}       & $O(n |\tx|)$                                                                                                         & $O(n^3)$ & -                                                                            & $O(\log n)$ & $O(n^2)$ $^*$ & duplicated tx buffers + TPKE $\diamond$                                                                                            & \cmark                                                                             & \xmark                                                                                              \\\rule{0pt}{5pt}
			PACE~\cite{zhang2022pace}              & $O(n^2 |\tx|)$                                                                                                       & $O(n^3)$ & -                                                                            & $O(\log n)$ & $O(n^2)$ $^*$ & duplicated tx buffers $\ddagger$ (hash-only)                                                                                            & \cmark                                                                             & \xmark                                                                                                 \\\rule{0pt}{5pt}
			DL~\cite{yang2021dispersedledger}    & $O(\kappa n |\tx|)$  $\dagger$                                                                                                & $O(n^3)$ & -                                                                            & $O(\log n)$ & $O(n^2)$ $^*$ & strong validity   $\dagger$                                                                                              & \cmark                                                                             & \xmark                                                                                    \\\rule{0pt}{5pt}
			FIN~\cite{duan2023practical}               & $O(n^2 |\tx|)$                                                                                                       & $O(n^3)$ & -                                                                            & $O(1)$  & $O(n^2)$ $^*$   & duplicated tx buffers $\ddagger$ (hash-only)                                                                                            & \cmark                                                                             & \xmark                                                                                       \\\hline\rule{0pt}{10pt}
			CKPS01~\cite{cachin2001secure}            & $O(n^3 |\tx|)$                                                                                                       & $O(n^2)$ & $O(n^3)$                                                                     & $O(1)$   & -   & duplicated tx buffers   $\ddagger$                                                                                                                   & \cmark                                                                             & \xmark                                                                                               \\\rule{0pt}{5pt}
			sDumbo~\cite{guo2022speeding}            & $O(n |\tx|)$  $\dagger$                                                                                                        & $O(n^2)$ & $O(n^3)$; $O(n^4)$  impl.{$^\S$}                                         & $O(1)$    & -  & duplicated tx buffers + TPKE $\diamond$                                                                                              & \cmark                                                                             & \xmark                                                                                                 \\\rule{0pt}{5pt}
			Tusk~\cite{tusk}              & $O(n^2 |\tx|)$                                                                                                       & $O(n^2)$ & $O(n^3)$; $O(n^4)$  impl.{$^\S$}                                         & $O(1)$  & -    & duplicated tx buffers     $\ddagger$                                                                                                                & \cmark                                                                             & \cmark                                                                                                 \\\rule{0pt}{5pt}
			Dumbo-NG~\cite{gao2022dumbo}          & $O(\kappa n |\tx|)$  $\dagger$                                                                                                 & $O(n^2)$ & $O(n^3)$; $O(n^4)$  impl.{$^\S$}                                         & $O(1)$   & -   & strong validity $\dagger$                                                                                                & \xmark                                                                             & \cmark                                                                                                 \\ \hline\rule{0pt}{10pt}
			FIN-NG (this work)           & $O(\kappa n |\tx|)$  $\dagger$                                                                                                 & $O(n^3)$ & -                                                                            & $O(1)$  & -    & strong validity  $\dagger$                                                                                               & \cmark                                                                             & \cmark                                                                                                  \\\rule{0pt}{5pt}
			$\jumbo$ (this work)            & $O(\kappa n |\tx|)$  $\dagger$                                                                                                 & $O(n^2)$ & $O(n^2)$                                                                     & $O(1)$   & -   & strong validity  $\dagger$                                                                                               & \cmark                                                                            &  \cmark                                                                                                  \\ \hline\rule{0pt}{5pt}
		\end{tabular}
	}
	
	\vspace{-2mm}
	{
		\begin{scriptsize}
			\begin{itemize}[leftmargin=4mm]
				\item[$^\ddagger$] ``Duplicated tx buffers'' mean that a transaction has to be sent to $f+1$ honest nodes' transaction buffers, otherwise the transaction might never be agreed and violate liveness. This is a typical problem causing quadratic   (amortized) communication cost in  protocols   directly built from asynchronous common subset (ACS) such as PACE and FIN \cite{duan2023practical,zhang2022pace,benor},  because the adversary can enforce ACS to constantly drop $f$ honest nodes' proposed transactions.
				\item[$\diamond$] HBBFT, BEAT and sDumbo warp ACS by threshold encryption,  amortizedly reducing    communication cost  to  $O(n)$ despite duplicated tx buffers.
				\item[$^\dagger$] ``Strong validity'' guarantees every transaction proposed by any honest node to  be eventually agreed. It allows de-duplication techniques \cite{mirbft,tusk} to send each transaction to only $\kappa$ {\em random} nodes (where $\kappa$ is a small parameter ensuring overwhelming probability to include an honest node in $\kappa$ random nodes).
				\item[$^\S$] Tusk, sDumbo and Dumbo-NG \cite{gao2022dumbo,tusk,guo2022speeding} have $O(n^3)$ authenticator complexity if using non-interactive threshold signatures to implement QCs, but   their   implementations actually choose ECDSA/EdDSA for concrete performance, resulting in $O(n^4)$ authenticator complexity. 
				\item[$^*$] Though each node   performs $O(n)$ number of erasure-code  decoding, this  reflects a computation cost of $\tilde{O}(n^2)$, as  each decoding has $\tilde{O}(n)$ operations.
			\end{itemize}
		\end{scriptsize}
	}
	
	\vspace{-0.55cm}
\end{table*}

\subsection{Our Contribution}

We systematically investigate how to overcome the scalability issues of existing asynchronous BFT consensus and answer the aforementioned question in affirmative. A couple of more scalable asynchronous   consensus protocols $\finng$ and $\jumbo$ are designed and experimentally demonstrated. In greater detail, our core contributions can be summarized as:

\begin{itemize}

	\item {\bf FIN-NG: adapting the existing signature-free  approaches   into  parallel broadcast and agreement}. We initially propose a more scalable signature-free asynchronous BFT consensus protocol $\finng$, by adapting the signature-free state-of-the-art  FIN into the enticing paradigm of concurrent broadcast and agreement.  
	In comparison to FIN, $\finng$ avoids the unnecessary detour to ACS and can (i) achieve significantly higher throughput because    concurrently running broadcast and agreement can more efficiently utilize bandwidth resources, 
	(ii) ensures all honest nodes' proposals to finally output, thus realizing strong liveness to mitigate the censorship threat in FIN, 
	(iii) avoid full-fledged reliable broadcasts as well as computationally costly erasure-coding  in normal path to further reduce latency.
	FIN-NG not only greatly improves the performance of FIN, but also already outperforms the signature-involved state-of-the-art protocols like $\dumbong$ in most deployment settings.
	


	\smallskip
	\item {\bf Comprehensive benchmarking and   optimizations towards  concretely faster   QC implementation}. We   assess the impact of implementing  QCs from various   digital signatures, including concatenated ECDSA/Schnorr signature \cite{schnorr1991efficient}, non-interactively half-aggregatable Schnorr signature \cite{chalkias2021non,chen2022half},   BLS multi-signature  \cite{boneh2004short,boneh2018compact} and BLS threshold signature \cite{boldyreva2003threshold}. 
	\yuan{Neither same to many previous  practitioners  \cite{li2023performance,gao2022dumbo,guo2022speeding,tusk,bft-smart,yin2018hotstuff-full} suggesting ECDSA/EdDSA for (concrete) computing efficiency nor similar to many theorists using    threshold   signatures for  communication efficiency \cite{cachin2001secure,cachin00,abraham2018validated,lu2020dumbo}}, we strictly suggest to implement QCs       by   BLS multi-signature in an overlooked ``batch-then-verify'' manner \cite{jannes2023beaufort,opticoin},  with adding a ``blocklist'' against performance downgrade by excluding malicious signers  while aggregating. This enjoys the short size of aggregate signatures  and   avoids their   costly verifications or aggregation. 
	Moreover, regarding the major authenticator bottleneck of sending $\bigO(n)$ QCs, we   explore the merit of  BLS multi-signatures to further compress  $\bigO(n)$ QCs on different messages,  alleviating the communication   of transferring  $\bigO(n)$ QCs by saving more than half in size.

	
	\smallskip
	\item {\bf $\jumbo$: even more scalable asynchronous BFT  instantiation with $O(n^2)$ authenticators and messages}. 
	Despite $\finng$ and our various QC optimizations, we still seek for  a more scalable solution  to asynchronous BFT   that can asymptotically reduce authenticator complexity and simultaneously attain quadratic message complexity. 
	As shown in Table \ref{tab:compare}, we achieve the goal by presenting $\jumbo$---a non-trivial scalable instantiation of $\dumbong$, by introducing the information dispersal technique \cite{lu2020dumbo} to reduce the number of multicasting $\bigO(n)$ QCs from $n$ to expected constant. $\jumbo$ asymptotically reduces authenticator overhead with preserving a balanced $O(n^2)$ message complexity, thus significantly outperforming the state-of-the-art protocols and $\finng$ in most settings (particularly in these settings with larger scales and restricted bandwidth).
	Moreover, we propose a  patch for $\jumbo$ to achieve a desired security   property of ``fairness'', thus preventing adversaries from controlling arbitrary portion of  transactions in the consensus output.

\begin{figure}[H] 
	\vspace{-0.4cm}
	\begin{center}
		\includegraphics[width=7cm]{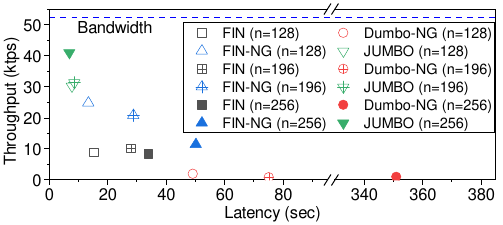}\\
		\vspace{-0.3cm}
		\caption{Our designs v.s.    state-of-the-art   in a WAN setting (400 ms RTT, 100 Mbps bandwidth) for $n=$128, 196 and 256.}	
		\label{fig:intro-latency}
	\end{center}
	\vspace{-0.25cm}
\end{figure}

	\item {\bf Open-source implementation and extensive evaluations in varying network settings with up to 256 nodes}. We implemented $\jumbo$ and FIN-NG in Golang, 
	 and also $\dumbong$ \cite{gao2022dumbo} and FIN \cite{duan2023practical} in the same language as by-product for fair comparison. 
	These  protocols were    comprehensively evaluated in WAN settings for $n=$ 64, 128, 196, 256 nodes.
	Malicious attacks such as $\mvba$ quality manipulation and sending of false signatures are also evaluated.
	Figure \ref{fig:intro-latency} highlights part  of  our results in a network of 100 Mbps bandwidth and 400 ms RTT. \footnote{Most cloud providers charge  $>$1,000 USD/month for reserving 100 Mbps bandwidth, and 400 ms RTT is   typical for cities between Asia and America. Such   condition   arguably reflects   a setting of affordable  global deployment.}
	 Noticeably, $\jumbo$ is   5X and 70X faster than FIN and $\dumbong$, respectively,  in the setting of 256 nodes.
	 More interestingly, $\jumbo$ can  realize the least latency and simultaneously achieve the maximal throughput closer to line rate.
	For detailed performance comparison among these protocols, cf. Section \ref{sec:eval}.

		Our open-source codebase of   $\jumbo$, $\finng$,      FIN  and $\dumbong$ is of about 13,000 LOC and available  at: 	
	\begin{center}
 
			\texttt{\url{https://github.com/tca-sp/jumbo}}. 
 
	\end{center}

\end{itemize}

\ignore{ 
\section{Challenges and Technique Overview}

Fast QC implementation with preserving security and worst-case performance.

Extension of the state-of-the-art with less overhead for larger scales

1/2-quality $\mvba$ protocol without signature-based quorum.
}


\section{Challenges and Our Techniques}

Bearing the major scalability bottleneck of authenticator complexity in existing signature-involved asynchronous BFT protocols,
our first attempt is to resolve the efficiency issues of signature-free prior art (like FIN), towards achieving a better scalable signature-free asynchronous BFT protocol $\finng$.
Then, we dedicatedly optimize the state-of-the-art signature-involved protocols  (like $\dumbong$), leading up to a scalable asynchronous BFT instantiation $\jumbo$ that not only asymptotically reduce  the authenticator overhead but also is concretely efficient in large-scale networks.

\begin{figure}[H]
	\begin{flushleft}
		\includegraphics[width=9cm]{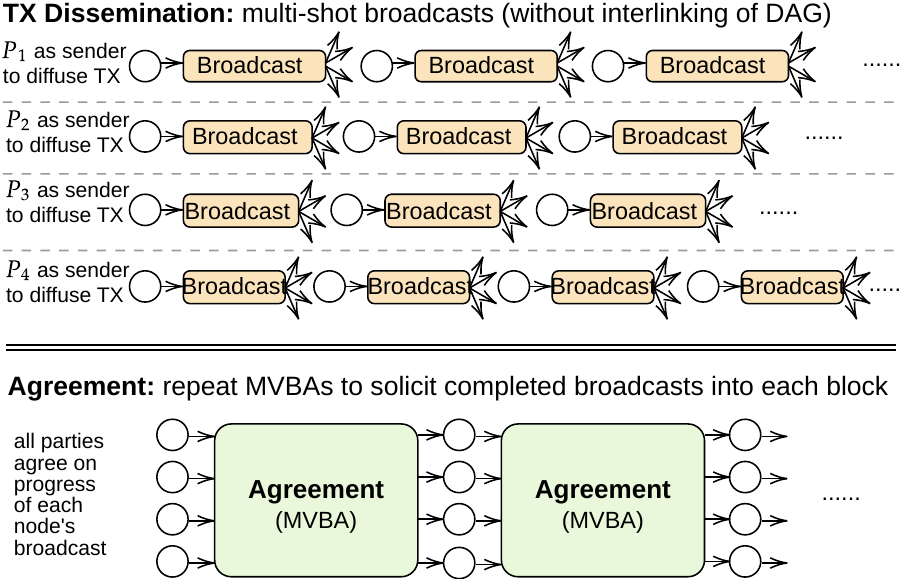}\\
		\vspace{-0.3cm}
		\caption{The paradigm of parallel broadcasts and agreement (the high-level rationale behind both $\finng$ and $\jumbo$).}	
		\label{fig:parallel}
	\end{flushleft}
\end{figure}
\vspace{-0.6cm}

\subsection{Challenges and Techniques of $\finng$}

As illustrated in Figure \ref{fig:parallel},
$\finng$ has a straightforward high-level idea similar to \cite{gao2022dumbo}, i.e., adapt the existing signature-free asynchronous common subset protocols like FIN into the enticing paradigm of parallel broadcasts and agreement.
The paradigm can deconstruct  asynchronous common subset protocol   into two concurrent processes: (i) one bandwidth-intensive process consists of $n$ multi-shot broadcast instances, and each broadcast allows a designed node to sequentially disseminate  its  transactions to the whole network, and (ii) the other agreement process  executes a sequence of asynchronous multi-valued validated Byzantine agreement ($\mvba$)   protocols, such that
every $\mvba$ can solicit  some latest completed broadcasts and decide a new block to output.

However, when adapting FIN into the enticing paradigm of parallel broadcast and agreement, we  face a couple of new challenges that do not exist in the signature-involved setting.



\smallskip
\noindent
{\bf \em  Challenge NG-1: erasure-code computation might bring heavy CPU cost in broadcasts}. 
The first problem of implementing $\finng$ is how to efficiently implement its broadcast primitives.
One might suggest    the   communication-efficient variants \cite{renavss,cachin05,alhaddad2022balanced} of Bracha's reliable broadcast \cite{bracha1987asynchronous,bracha1984asynchronous}. However, although these  protocols can attain  amortized linear communication cost (for sufficiently large input batch), their price is   the necessity of erasure-coding.
What is worse, the efficiency issue of using erasure-code in these broadcast protocols could be very serious, even if we cherry-pick fairly fast  implementation as suggested by \cite{beat} (cf. Figure \ref{fig:fin-break}), because these broadcast protocols rely on sufficiently large input batches to amortize their communication cost, making
encoding over very large fields and thus unsurprisingly   slow.

{\bf \em Our approach}. We overcome the efficiency issue by introducing a weakened variant of reliable broadcast   that can be efficiently implemented without any erasure-coding.
Similar to the normal path of PBFT \cite{pbft}, the broadcast primitive no longer ensures all honest nodes to eventually receive the broadcast value; instead, it only guarantees $f+1$ honest nodes to receive the original broadcast value, and   allows the remaining honest nodes only eventually receive a hash digest of value.
We further provide a pulling mechanism to compensate the weakening of broadcast. The pulling mechanism incorporates dispersal technique \cite{cachin05} for communication optimization, such that for any honest node  that only receives the hash digest, it can efficiently fetch  the missing transactions from other honest nodes, at an amortized linear communication cost.

\smallskip
\noindent
{\bf \em Challenge NG-2:   poor ``quality'' of existing signature-free MVBA might hinder liveness of { }$\finng$}.
It is notable that in $\dumbong$, its liveness relies on the so-called ``quality'' of $\mvba$ protocols.
Here quality means the output of $\mvba$ shall not be fully manipulated by the adversary.
Imagine that the adversary can completely manipulate the $\mvba$ result, it thus can   rule out a certain honest node's broadcast from the finalized output (because the $\mvba$ result determines which nodes' broadcasts shall be solicited by the output block).
This clearly violates   liveness, because some honest node's broadcasted transactions wouldn't eventually output anymore.
In  the signature-based setting, there are several efficient $\mvba$ protocols \cite{guo2020dumbo,abraham2018validated,lu2020dumbo} with optimal 1/2 quality, 
such that $\dumbong$ can be securely instantiated with preserving liveness. However, to our knowledge, there is no  signature-free $\mvba$ protocol with satisfying     quality property, which might potentially cause liveness vulnerability in $\finng$.

{\bf \em Our approach}. To deal with the potential liveness issue, we necessarily redesign the state-of-the-art signature-free $\mvba$  protocol FIN-$\mvba$ \cite{duan2023practical} to ensures   optimal 1/2 quality \cite{goren2022probabilistic} against adaptive adversaries, i.e. the output is proposed by some honest node with at least 1/2 probability. In contrast, FIN-$\mvba$ protocol only ensures 1/3 quality in the same setting.\footnote{Note that although \cite{duan2023practical} proposed a quality patch for its $\mvba$, it still only ensures 1/3 quality against an adaptive   adversary with ``after-fact-removal'' (i.e., the adversary can adaptively corrupt some node and remove this node's undelivered messages), cf.  Appendix \ref{app:quality-attack} in Supplementary for   detailed analysis.} The seemingly small difference  in $\mvba$'s quality can reduce the   latency by an entire $\mvba$ execution, and can render a considerable improvement of latency in $\finng$ under Byzantine setting (cf. Section \ref{sec:eval} for our experiment evaluation of $\mvba$ quality's impact on latency).

\subsection{Challenges and Techniques of $\jumbo$}

As a concrete instantiation of $\dumbong$, the design of $\jumbo$ also follows the high-level structure of parallel broadcast and agreement shown in Figure \ref{fig:parallel}.
However, to make $\jumbo$ practically operate in large-scale networks with enhanced fairness guarantee, we necessarily  overcome the following challenges to reduce the prohibitive authenticator complexity in $\dumbong$, both    concretely and asymptotically.

\smallskip
\noindent
{\bf \em Challenge J-1: trivial use of aggregatable signatures like BLS could be computationally slow}.
Although it is well known that   aggregatable   signatures like BLS threshold signature 
can make great theoretic improvement to BFT protocols,  it is not a popular choice   in real-world systems.
In particular, some very recent BFT protocols like HotStuff \cite{yin2018hotstuff-full}, $\tusk$ \cite{tusk}  and $\dumbong$ \cite{gao2022dumbo}   choose ECDSA/EdDSA signatures instead of BLS threshold signature to implement  QCs in practice.
And a recent benchmarking \cite{li2023performance} even suggests ECDSA/EdDSA signature always favorable than BLS threshold signature for large-scale networks,
as it points out that (the trivial use of) BLS threshold signature might suffer from huge computation costs like the verification of individual signature and the aggregation involving interpolation in exponent.

{\bf \em Our approach}. To make QC  implementation concretely faster, we provide a resolution to settle the debate regarding   QC implementation between theory and practice,
and affirm that  QC from BLS multi-signature (with non-trivial optimizations) could be the best practice for large-scale BFT systems.
First, BLS multi-signature has  very simple and efficient aggregation, which is concretely much faster than its threshold counterpart as it does not  perform the heavy computation of interpolation in exponent.
Second, we extend the optimistic ``batch verification''  approach for BLS multi-signatures \cite{opticoin,jannes2023beaufort}
by adding an easy-to-implement ``blocklist'' mechanism to rule out malicious nodes from quorum, thus
amortizedly extending  its superior efficiency in the optimistic case   into the worst case.
Third, BLS multi-signature allows to aggregate $O(n)$ QCs across various messages (i.e., the primary authenticator bottleneck  in $\dumbong$ and $\tusk$), and we leverage this nice property to further reduce the concrete size of $n$ different QCs  by  more than half.


\smallskip
\noindent
{\bf \em Challenge J-2: how to  achieve  asymptotically lower authenticator complexity?  By extending $\dumbong$ or $\tusk$?}
It is worth noticing $\tusk$ has a directed acyclic graph (DAG) structure where every node has to multicast at least $n-f$ QCs of preceding broadcasts,
These QCs are so-called ``edges'' pointing to preceding broadcasts to form DAG.
Therefore, all multicasts of $n-f$ QCs are intuitively needed in $\tusk$  as part of its DAG structure.
In contrast, a key observation of $\dumbong$ is that the cubic term of its authenticator complexity  has a unique reason caused by the input multicast phase of $\mvba$ protocols. 
This fact raises a question to us:   Can we   reduce the number of input multicasts inside $\dumbong$'s $\mvba$ protocol? This is probably plausible because $\mvba$ only decides one node's input as output, so   the remaining input multicasts are essentially redundant.

{\bf \em Our approach}. We  leverage the neat structure of $\dumbong$   to incorporate a simplistic manner to   reduce the authenticator complexity by an $\bigO(n)$ order. In short, we apply $\mvba$ extension protocol \cite{lu2020dumbo} to replace the multicasts of $n$ QCs by more efficient   dispersals of QCs.  More precisely, we let every node use a provable dispersal protocol \cite{lu2020dumbo} to efficiently disperse $n$ QCs (instead of trivially multicasting), then only a constant number of dispersals (instead of all $n$ dispersals) are  reconstructed for verifications. This dispersal-then-recast paradigm reflects  our observation that only one   multicast of $n$ QC is necessary, and all else are essentially redundant. Clearly, this optimized execution flow is asymptotically more efficient, because each dispersal   has a much less commutation cost  that  is   only $O(1/n)$ of  trivial multicast. 
\section{Problem Formulation  and Preliminaries}

\subsection{System and Threat Models}\label{sec:model}

We adopt a widely-adopted asynchronous message-passing model with initial setups and Byzantine corruption \cite{cachin00,cachin2001secure,abraham2018validated,lu2020dumbo}, which can be formally described in the following:

\smallskip
\noindent
{{\bf Known participating nodes}}. There are $n$ designated     nodes in the system. We denote them by $\{\node_i\}_{i \in [n]}$, where $[n]$ is short for $\{1, \dots, n\}$. Each   $\node_i$ is bounded to some public key $pk_i$   known by all other nodes through a bulletin board PKI (though we might not use the setup in the signature-free construction).


\smallskip
\noindent
{{\bf Established threshold cryptosystem}}. To overcome   FLP impossibility, we  depend on  a cryptographic coin flipping protocol \cite{cachin00}, which can return unpredictable common randomness upon a threshold of participation nodes invoke it. 
This usually requires some threshold cryptosystem \cite{libert2011adaptively,cachin00} that is already setup, which can be done through either asynchronous distributed key generation   \cite{kokoris2020asynchronous,das2022practical,gao2021efficient,abraham2021reaching} or  a trusted third party.

\smallskip
\noindent
{{\bf Computationally-bounded Byzantine corruption}}. We   focus on   optimal resilience in asynchrony against an adversary $\adv$ that can corrupt up to $f=\lfloor (n-1)/3 \rfloor$   nodes. 
Following standard cryptographic practice, the adversary  is probabilistic polynomial-time bounded, making it infeasible to break  cryptographic primitives. 
We might consider static or adaptive adversaries, where static adversaries choose a fixed set of $f$ nodes to corrupt  before the protocol starts, and adaptive ones can gradually corrupt   $f$ nodes  during    protocol execution.

\smallskip
\noindent
{{\bf Reliable authenticated asynchronous network}}. We consider a reliable asynchronous message-passing network with fully-meshed point-to-point links and authentications. Messages sent through the underlying network can be arbitrarily delayed by the adversary due to asynchrony, but they can be eventually delivered  without being tampered or forged. In addition, while a node receives a message, it can assert the actual sender of the message due to authentication. 


	\vspace{-0.2cm}

\subsection{Design Goal: Asynchronous BFT Consensus  with Fairness}

We aim at  secure asynchronous BFT consensus protocols satisfying the next atomic broadcast  abstraction.
 
\noindent
\begin{definition}{\bf Asynchronous BFT atomic broadcast (ABC).}\label{def:abc} In an ABC protocol, 
	each node has an input buffer of transactions, and continually outputs some blocks of transactions that   form an ever-growing and append-only linearized log of transactions. 
	The protocol shall satisfy the next properties with all but negligible probability:
 
	 \begin{itemize}
	 	\item Agreement. If an honest node outputs a transaction, then all honest nodes   eventually output it.
	 	\item Total Order. For any two output logs of  honest node(s), one is a prefix of the other, or they are equal. 
	 	\item Liveness (Strong Validity). If a transaction is placed in any honest node input buffer, it can eventually be included by the output log of some honest node.
	 	\item $\alpha$-Fairness. 
	 	For any block  outputted by some honest node, if the block contains $K$ transactions,
	 	then at least $\alpha \cdot K$ transactions in the block are proposed by honest nodes.
	 \end{itemize}
\end{definition}
 
Some remarks on the above definition:

\begin{itemize}



\item For liveness (validity), we define it in a   strong form: if a transaction stays in any honest node's input buffer, it can eventually output, while some weaker forms only guarantee this if the transaction is input of $f+1$ honest nodes or all honest nodes. Strong validity is a useful property enabling    de-duplication of input transactions  \cite{gao2022dumbo,mirbft,tusk}, namely, allowing a client forward its transactions to only a small number of $k$ consensus nodes. 

\item Fairness   bounds the ``quality'' of consensus  \cite{garay2015bitcoin} and prevents   the adversary   controlling an arbitrarily large portion in the output, ensuring that  the consensus output contains sufficiently many transactions   proposed by honest nodes. 

\end{itemize}


\subsection{Complexity Measurements}\label{app:complexity}

Besides   security, we are also interested in constructing practical asynchronous BFT protocols.
For the purpose, we consider the next complexity  metrics as efficiency indicators:

\begin{itemize}[leftmargin=6mm]
	\item {\em Message complexity} \cite{cachin2001secure}: the expected  number of messages   sent by honest nodes to decide an output. Note that we do not count message complexity amortized for batched decision in order to  reflect the actual overhead of packet/datagram headers in a real networked systems.
	
	\item {\em (Amortized) bit communication complexity} \cite{cachin2001secure}: the expected number of bits sent by honest nodes per output transaction. Note that bit communication complexity is amortized for a batch of decision to reflect the number of communicated bits related to each output transaction.
	
	\item {\em Round complexity}  \cite{canetti1993fast}: the expected  asynchronous rounds     needed  to  output a transaction. 
	{Here asynchronous round 
		is a ``time'' unit defined by the longest delay of messages sent among honest nodes \cite{canetti1993fast,dag}.}

	\item {\em Authenticator complexity}  \cite{yin2018hotstuff-full,avarikioti2020fnf}: the expected communicated bits generated by honest nodes associated with transferring (or verifying) signatures  per output decision. For protocols that all have the optimal (linear) amortized bit communication complexity, authenticator communication complexity can provide a more fine-grained measurement to distinguish their actual communication overheads, thus guiding us towards more scalable protocol  designs.


\end{itemize}

\subsection{Preliminaries}\label{sec:protocols}

\noindent \textbf{Byzantine reliable broadcast ($\rbc$)}. A $\rbc$ protocol \cite{bracha1987asynchronous,cachin2005asynchronous} among $n$ participating nodes has a designated sender and can simulate an ideal broadcast channel in an point-to-point network, thus realizing 
(i) {\em agreement}: any two honest nodes' outputs are same;
(ii) {\em totality}:    all honest nodes  would {\em eventually} output some value  conditioned on some honest node has output, 
and (ii) {\em validity}:  an honest sender's input can {\em eventually} be  output by all honest nodes.

\smallskip
\noindent \textbf{Weak Byzantine reliable broadcast (w$\rbc$)}. A w$\rbc$ protocol is a relaxed $\rbc$, where  an honest node    either outputs a value or a hash digest. The protocol satisfies 
(i) {\em weak agreement}: two honest nodes' outputs are either the same value, or the same hash digest, or a value and its hash digest; 
(ii) {\em weak totality}:  at least $n-2f$ honest nodes would eventually output some value and all rest honest nodes would eventually output some  digest conditioned on some honest node has output,
(ii) {\em validity}: same to $\rbc$.

\smallskip
\noindent \textbf{Multi-valued validated Byzantine agreement ($\mvba$)} \cite{abraham2018validated,cachin2001secure,lu2020dumbo}.
This is a variant of Byzantine agreement with external validity, such that all honest  nodes can agree on a value satisfying a publicly known predicate $Q$.  More formally,

\begin{definition}
	Syntax-wise, each node in the $\mvba$ protocol   takes a (probably different) value validated by a global predicate $Q$ (whose description is known by the public) as input, and decides a   value satisfying $Q$ as the output. 
	%
	An $\mvba$ protocol running among $n$ nodes shall satisfy the next properties except with negligible probability:
	\begin{itemize}[leftmargin=6mm]
		\item {\em Termination.} If all honest nodes  input some values (eventually) satisfying $Q$, then each honest node would output;
		
		\item {\em Agreement.} If  two honest nodes output $v$ and $v'$, respectively, then $v=v'$.
		
		\item {\em External-Validity.} If an honest node outputs a value $v$, then $v$ is valid w.r.t. $Q$, i.e., $Q(v) = 1$;	
		
		\item {\em 1/2-Quality.} If an honest node outputs $v$, the probability that $v$ was input by the adversary is at most 1/2.	
		
	\end{itemize}
\end{definition}
 
\noindent

In addition,  many studies \cite{yurek2023long,das2023practical,abraham2023perfectly,duan2023practical} recently realized that the external validity property of $\mvba$ can be extended to support a somewhat {\em internal} predicate, which does not return ``false'', but keeps on listening the change of the node's local states and returns ``true'' only if the joint of predicate input and a node's local states satisfies a pre-specified validity condition. This is particularly useful if the higher level protocol  can enforce all nodes' internal predicates to eventually be true as long as any honest node's predicate is true, e.g., all nodes keep on checking whether they output in a specific $\rbc$ protocol. Later in Section \ref{sec:jumbo}, the {\em internal} predicate allows us to design FIN-NG by adapting the signature-free ACS protocol FIN into the paradigm of concurrent agreement and broadcast.

\smallskip
\noindent \textbf{Reproposable asynchronous binary agreement ($\raba$)} \cite{zhang2022pace}. This is a special variant of binary Byzantine agreement, where each node is additionally allowed to change its input from 0 to 1 (but not vice versa).  Here we recall its   definition:
 
\begin{definition}
	Syntactically, a node in $\raba$ is allowed to invoke two input interfaces $raba$-$propose$ and $raba$-$repropose$, and has an output interface $raba$-$decide$ to receive the output;  an honest node is required to use the   interfaces in the following way: (i) must invoke $raba$-$propose$ before $raba$-$repropose$, (ii) cannot use any interface more than once, (iii) can only take 0 or 1  as input of $raba$-$propose$ and $raba$-$repropose$ interfaces, (iv) cannot $raba$-$propose$ 1 and then $raba$-$repropose$ 0 (but   vice versa is allowed). A $\raba$ protocol shall satisfy the next properties except with negligible probability:
	\begin{itemize}[leftmargin=6mm]
 
		\item {\em Agreement.} If any two honest nodes terminate, then they $raba$-$decide$ the same value.
		\item {\em Unanimous termination and Unanimous validity.}  If all honest nodes $raba$-$propose$ $v$ and never $raba$-$repropose$, then all honest nodes terminate and $raba$-$decide$ $v$.
		\item {\em Biased validity.} If $f+1$ honest nodes $raba$-$propose$ 1, no honest node would $raba$-$decide$ 0.
		\item {\em Biased termination.} If all honest nodes $raba$-$propose$ 1 or eventually $raba$-$repropose$ 1, then all honest nodes would terminate and $raba$-$decide$.
	\end{itemize}
\end{definition}

 
\noindent
{\bf Cryptographic notions}.
$\hash$ denotes a collision-resistant one-way hash function. $\Sign$ and $\Verify$ represent the signing and verifying algorithms of a  signature scheme that is existentially unforgeable under adaptive chosen message attack (i.e., EUF-CMA secure). We also consider non-interactively aggregatable multi-signature (or threshold signature) \cite{chen2022half,chalkias2021non,boldyreva2003threshold,boneh2004short,boneh2018compact}, where is another $\Agg$ algorithm that can input multiple signatures on a message and output a  multi-signature (or threshold signature) on the same message. 

\smallskip
\noindent
{\bf Aggregating BLS (multi-)signatures in KOSK model}. 
It is noticeably that we adopt the standard knowledge-of-secret-key (KOSK) model \cite{boldyreva2003threshold}:  every   node has to sign a valid  signature to prove its knowledge of secret key when registering public key. This is the standard process of real-world CA while issuing certificates \cite{rfc2986} and also reflects the public key registration  in proof-of-stake  blockchains.
The KOSK model brings very convenient aggregation of BLS  signatures. In particular,  the simplistic aggregation is in form of $\Agg(\sigma_1, \cdots ,\sigma_k)=\prod_{j=1}^k\sigma_j$, which simply returns $\sigma_1\cdots\sigma_k$ as result. Here $\sigma_1, \cdots ,\sigma_k$ can be BLS signatures or even already-aggregated BLS multi-signatures on either same or different messages.
Later in Section \ref{sec:qc}, we will showcase how asynchronous BFT can leverage this   aggregation technique towards aggregating $\bigO(n)$ QCs on different messages.

%
%

\section{Initial attempt: $\finng$ to scale up Signature-free   Asynchronous BFT}\label{sec:jumbo}

This section   attempts to enhance the scalability of   signature-free asynchronous BFT protocols,
by adapting the cutting-edge signature-free ACS protocol FIN into the enticing paradigm of concurrent broadcast and agreement. 
The resulting protocol $\finng$ already outperforms   existing asynchronous BFT protocols in the bandwidth-sufficient setting.
Finally, we explain why cubic messages of $\finng$ might cause unsatisfying performance degradation in the bandwidth-starved environment.


\begin{figure}[H] 
	\vspace{-0.1cm}
	\begin{center}
		\includegraphics[width=9.1cm]{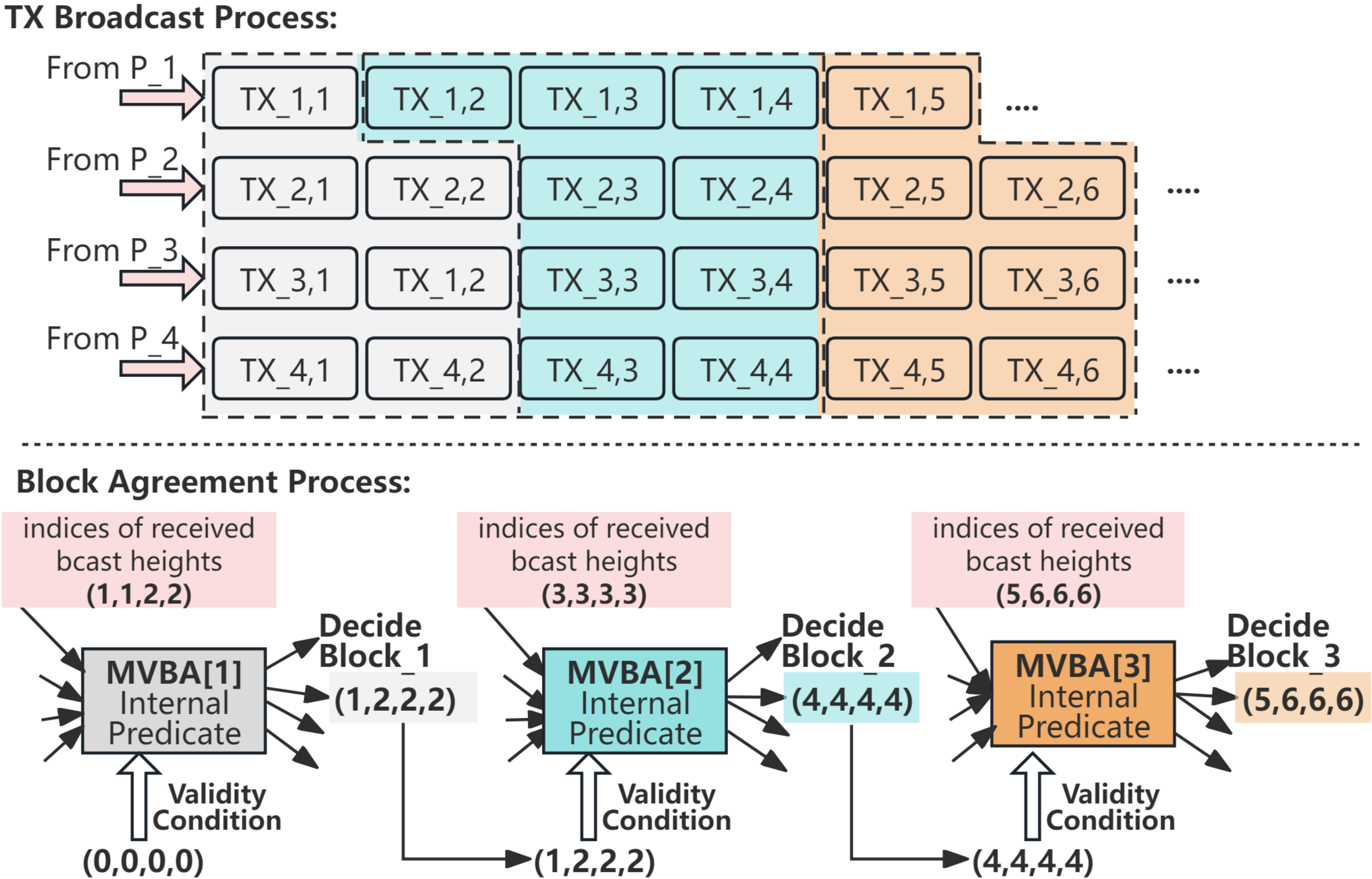}\\
		\vspace{-0.1cm}
		\caption{High-level of $\finng$.}	
		\label{fig:fin-ng}
	\end{center}
	\vspace{-0.5cm}
\end{figure}

\smallskip
\noindent{\bf Overview of $\finng$}. As shown in Figure \ref{fig:fin-ng},   $\finng$ deconstructs  FIN ACS \cite{duan2023practical} to separate transaction broadcast from consensus. 
Noticeably, a trivial parallelization of FIN ACS's broadcast and agreement wouldn't bring us  well-performing  $\finng$ in large-scale networks, 
and we specially tailor  both modules  to alleviate their efficiency bottlenecks  as follows:

\begin{itemize}
	\item {\bf Avoid costly $\rbc$ towards more efficient broadcast}: 
	  FIN instantiates its broadcast process by using  CT05 $\rbc$   \cite{cachin2005asynchronous}. The protocol leverages erasure-code     to attain amortized linear communication cost (for sufficiently large batch of broadcast input).
	However, CT05 $\rbc$   involves expensive decoding, which could be very costly in FIN as each node has to participate  in $n$ such $\rbc$s and perform $n$ times of decoding. 
	We thus suggest to employ a more efficient process of transaction dissemination with weakened reliability.
	 That means, only $f+1$ honest nodes are ensured to receive consistent transactions and up to $f$ honest nodes might only receive the hash digest  (which can be analog to   PBFT \cite{pbft,bft-smart}). Thanks to such weakening, the broadcast attains amortized linear communication cost in the absence of erasure coding.
	To compensate   the weakening in broadcasts, we further propose a couple of communication-efficient  pulling mechanisms, such that any honest node  can securely fetch its missing transactions with amortized linear communication.

	\item {\bf Enhance $\mvba$ quality for liveness and lower latency}: 
	When separating broadcasts and $\mvba$ toward realizing $\finng$, the liveness  might be hurt if the underlying $\mvba$ has sub-optimal quality.
	This is because 
	if the adversary can   manipulate $\mvba$'s output 
	due to low quality,  it can make $\mvba$ result never contain the indices of some   nodes' completed  broadcasts, thus censoring certain nodes and breaking liveness.
	%
	To guarantee liveness and fast confirmation in $\finng$, we redesign the signature-free $\mvba$ protocols proposed in FIN \cite{duan2023practical}. 
	The resulting  FIN-$\mvba$-Q protocol can ensure optimal 1/2 quality in the influence of a strong adaptive adversary with ``after-the-fact-removal''. 
	In contrast,    $\mvba$ protocols in FIN \cite{duan2023practical} only realizes $1/n$ quality  and $1/3$ quality, respectively, 
	in the same setting. 
	Later in Section \ref{sec:eval}, we experimentally test   $\finng$  using $\mvba$s with different quality,
	revealing that our seemingly small improvement in $\mvba$ quality can make a large difference in latency.
	
\end{itemize}


\subsection{$\finng$ Protocol}
\label{sec:jumbo1}

This subsection   details   $\finng$ by
elaborating its      transaction dissemination process and block agreement process.

\smallskip
\noindent{\bf Broadcast process in $\finng$}.
Figure \ref{fig:broadcast} illustrates the   broadcast process that disseminates transactions in $\finng$.
There could be $n$ such multi-shot broadcasts in $\finng$,
where each node $\node_i$   acts as a  leader in one broadcast instance, 
and all nodes   execute as receivers in all  $n$ broadcasts. 
A multi-shot broadcast instance in $\finng$ is executed as follows.

\begin{itemize}
	\item 	{\em Tx dissemination through a  sequence of weak $\rbc$}. 
	The designated sender $\node_s$ and all other nodes participate in a sequence of weak reliable broadcast protocols (w$\rbc$ defined in Sec. \ref{sec:protocols}).
	The broadcast process    resembles a simplistic normal-path of PBFT,
	so it only ensures that at least $f+1$ honest nodes might (eventually) receive the original broadcasted transactions,	
	and up to $f$ else honest nodes might only receive a consistent hash digest. 
	
	To bound the memory of buffering received transactions, all broadcast participants can dump   transactions into persistent storage if some w$\rbc$ outputs them, which is because $\finng$   ensures strong validity to guarantee all broadcasted transactions to be part of consensus output (so we can write these broadcast results into disk).

	\item 	{\em Pulling mechanism in companion with broadcasts}.
	We can   enable  the $f$ honest nodes that only hold hash digest   to fetch missing transactions from other honest nodes.
	Since   the node still can  obtain  the hash digest (despite missing the original transactions), it can  securely fetch the corresponding   transactions, because (i) it knows which w$\rbc$ to fetch due to the receival of digest,
	and (ii) at least $f+1$ honest nodes can   output the corresponding transactions due to the (weaker) totality property of w$\rbc$. 	
	%
	Given this, we can propose two choices of implementing  the efficient pulling mechanism:
	\begin{itemize}
		\item The first approach  resembles Tusk's pull mechanism  against static  adversaries \cite{tusk}: when   $\node_i$ receives a hash digest from a certain w$\rbc$ instead of the original transactions, it  asks $\kappa$ random nodes for the missing transactions. There shall be an overwhelming probability such that these selected nodes would contain at least one honest node receiving the original transactions, implying that $\node_i$ can eventually receive at least one respond carrying the correct transactions consistent to its hash digest, with except negligible probability   in $\kappa$. 

		\item The second way is secure against adaptive adversaries \cite{gao2022dumbo}. 
		It lets  $\node_i$ ask all nodes when $\node_i$  realizes some missing transactions, with preserving
		  {\em amortizedly} linear communication complexity per decision.
		The crux of improving communication efficiency is: the responses to $\node_i$ are no longer the original transactions, 
		but incorporate communication-reducing techniques from asynchronous information dispersal \cite{cachin2005asynchronous,renavss}.
		For each node $\node_j$, it chops the  transactions requested by $\node_i$ into $f+1$ fragments, and uses erasure code   to encode the $f+1$ fragments into $n$ code fragments, then computes a Merkle tree using   $n$ code fragments as leaves, 
		and sends $\node_i$ the Merkle tree root, the $j$-th code fragment, and the corresponding Merkle proof attesting   the $j$-th code fragment's inclusion in the Merkle tree.
		At the $\node_i$ side, it can  finally receive  at least $f+1$ responses carrying the same Merkle tree root with valid Merkle proofs, such that 
		$\node_i$ can decode the $f+1$ code fragments   to recover the missing transactions. 

	\end{itemize}

\end{itemize}

\begin{figure}[H] 
	\begin{center}
		\includegraphics[width=9.9cm]{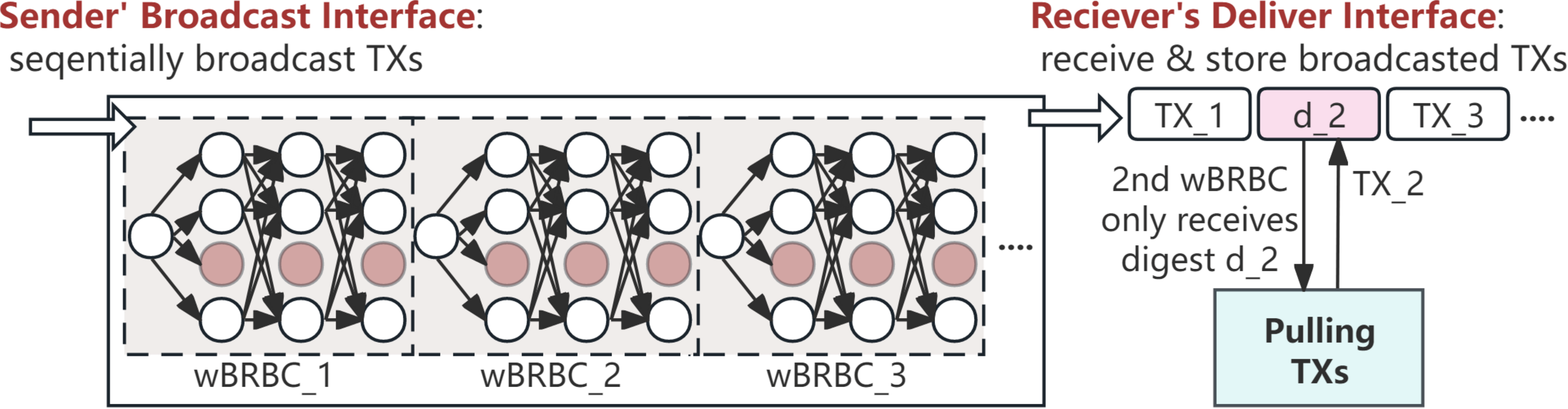}\\
		\vspace{-0.2cm}
		\caption{Broadcast instance in $\finng$.}	
		\label{fig:broadcast}
	\end{center}
	\vspace{-0.3cm}
\end{figure}

\smallskip
\noindent{\bf Agreement process in $\finng$}. 
$\finng$'s agreement process   is a sequence of $\mvba$s proceeding as follows by epoch $e$:
	
	When entering epoch $e=1$, each node $\node_i$ initializes a $n$-size vector $\ordered_e =[0,\cdots,0]$. 
	    $\node_i$ also maintains a $n$-size vector $\height_i$,
	where the $j$-th item closely tracks that $\node_i$ has received how many w$\rbc$s broadcasted by $\node_j$.
	When $\height_i$ has   $n-f$ items larger than $\ordered_e$ and has all items not smaller than $\ordered_e$ in the corresponding positions, 
	the node $\node_i$ takes a snapshot of $\height_i$ as $\mvba[e]$'s input. 
	$\node_i$ waits for that $\mvba[e]$ returns a vector $\newordered$ consisting of  $n$ indices, and can pack all w$\rbc$s with indices between $\ordered_e$  and $\newordered$ as a block of consensus result.\footnote{If    $\mvba[e]$   decides a certain w$\rbc$ as consensus output but $\node_i$ has not yet received anything from the w$\rbc$, the node can also invoke the pulling mechanism of broadcast instance to fetch the missing transactions.}
	Finally, $\node_i$ sets $e = e + 1$ and $\ordered_{e} = \newordered$, then enters the next epoch $e$.

	Noticeably, $\mvba$ used in the above procedures shall be specified by a predicate function satisfying the next conditions:
	\begin{enumerate}
		\item Check that each item in $\height_j$ is not {\em shrink} w.r.t. the last output of $\mvba$ (i.e. $\ordered_{e}$), return false on check fails, otherwise go the next step.
		\item Check that at least $n-f$ items in $\height_j$ increases w.r.t. the last output of $\mvba$, return false on check fails, otherwise go the next step.
		\item Wait for that each item in the local state $\height_i$ is larger than or equal to the corresponding item in  $\height_j$, and then return true. 
	\end{enumerate}

Similar to \cite{yurek2023long,das2023practical,abraham2023perfectly}, 
the above conditions requires some local states (i.e.,  $\height_i$----each broadcast's progress from $\node_i$'s view) to verify. 
One might worry its   step 3) might wait forever and hurt termination of $\mvba$.
However, violation of termination wouldn't happen because 
(i) any honest node $\node_i$'s input to $\mvba$ must be valid regarding its own local states;
(ii) the totality of broadcasts therefore ensure all honest nodes can eventually update their local  states to validate $\node_i$'s input.




\smallskip
\noindent{\bf Security of $\finng$}. The security of $\finng$ can be summarized as the following main theorem:
 
\begin{theorem}\label{thm:finng}
	$\finng$  securely realizes asynchronous BFT   atomic broadcast without fairness  except with negligible probability, conditioned on the underlying w$\rbc$ and $\mvba$ protocols are secure (where $\mvba$ shall satisfying quality).
\end{theorem}
 
\noindent
{\bf Security analysis.} The  rationale  behind Theorem \ref{thm:finng} is straight:
\begin{itemize}
	\item For safety, the broadcast   of $\finng$ built from w$\rbc$ can prevent malicious nodes to broadcast different transactions to distinct nodes due to its weak agreement; 
	moreover, $\mvba$'s validity and agreement can further ensure all honest nodes would solicit the same finished broadcasts   into each output block, thus ensuring safety.
	\item For liveness, w$\rbc$'s validity ensures that all honest nodes can reliably diffuse their transactions in only constant rounds; moreover, $\mvba$'s termination and quality gauntness that once an honest node's broadcast is delivered to all other honest nodes, this broadcast (or the same node's higher broadcast) would be solicited by some $\mvba$ output with an overwhelming probability of $1-(1-q)^k$ after $k$ $\mvba$ executions, where $q$ is the quality of   $\mvba$ protocol, and $(1-q)^k$ represents the probability that the adversary successively manipulates $k$ $\mvba$s' outputs to prevent the certain broadcast from outputting. Considering that the expectation of $k$ is  $1/q$ and $\mvba$ also has expected constant round complexity, any transaction that is broadcasted by some honest node would be decided as consensus output in expected constant rounds (i.e., the broadcast rounds plus $1/q$ times of $\mvba$'s expected rounds), thus realizing liveness.
\end{itemize}

We defer the full proof   to Appendix \ref{app:finng-proof} in Supplementary.

\ignore{
\begin{proof}[Proof Sketch]
	The roadmap of proving that $\finng$ satisfies all atomic broadcast properties is as follows.
			\yuan{pick up here to rewrite proof intuition....}
	\begin{itemize}

		\item For each item in the output of the current $\mvba$ protocol, all honest nodes shall deliver in the corresponding broadcast, due to the validation rules and agreement of $\mvba$ predicate; 
		moreover, out of these honest nodes, a minimum of $n-2f$ honest nodes can eventually deliver the same transaction batches, while the else honest nodes might only output the corresponding hash digest, because the weak agreement and totality of $w\rbc$. Nevertheless, there exist a request-and-pull mechanism  \cite{gao2022dumbo,tusk} to allow all honest nodes to fetch the  missing transactions according to hash digest, further ensuring the outputs of all honest nodes eventually consistent, thus proving total order and agreement. Proving the validity property is straightforward, reducing to the validity of $w\rbc$ and the quality of $\mvba$. For a detailed proof, interested readers can refer to \cite{gao2022dumbo}.
		
		\item Fairness  proof aligns with that in Theorem \ref{theorem:ng++} as the same fairness patch   can be adapted to $\jumbo$.
	\end{itemize}
\end{proof}
}

\smallskip
\noindent
{\bf Complexity analysis}. It is fairly simple to count the complexities of $\finng$ for its modular design. In particular, if $\finng$ is instantiated by  signature-free  $\mvba$ with quality and cubic messages (e.g., FIN-$\mvba$-Q to be introduced in next subsection), its authenticator complexity  is (nearly) zero, but at a price of  cubic messages per decision. 


\subsection{FIN-$\mvba$-Q: $\mvba$ with Optimal  Quality}

\noindent
If we prefer a signature-free instantiation of  $\finng$, 
it comes to our attention that   few satisfying   signature-free $\mvba$ protocols exist.
One candidate is the recent design of FIN-$\mvba$ \cite{duan2023practical}, which unfortunately  suffers from  quality degradation (cf. Appendix \ref{app:quality-attack} for detailed analysis) to cause the adversary hold 2/3 chance to manipulate   $\mvba$ output.
Worse still, quality of $\mvba$ is critical in $\finng$ as it is necessary for liveness:
the adversary can lower $\mvba$ output's quality   to slow down $\finng$'s confirmation or even completely censor certain nodes.
In order to guarantee   fast confirmation and liveness of $\finng$, 
we redesign  FIN-$\mvba$  to improve  its quality toward   1/2  against strong adaptive adversaries. Remarkably, our result   attains the optimal quality in asynchrony \cite{goren2022probabilistic}.

\smallskip
\noindent{\bf Abandonable validated Byzantine reliable broadcast}. 
We   identify the root reason of quality degradation in FIN-$\mvba$: an   adaptive adversary can hold a corruption quota until   coin flipping is revealed to determine the elected node, after which it can rushingly corrupt  this elected node and replaces the input value by the manipulating one. 
To prevent such rushing adversaries, we set forth a notion of abandonable validated  weak Byzantine reliable broadcast  (avw$\rbc$), such that the honest nodes can   reject a rushing adversary's value by abandoning the slowest broadcasts as soon as  the elected node is disclosed. 
Formally, avw$\rbc$ can be   defined as follows:

\begin{definition}
	An avw$\rbc$ protocol is a generalized w$\rbc$ protocol   with a validity predicate $Q()\rightarrow\{0,1\}$ for validating input. It also provides each   node a local $abandon()$ interface. An avw$\rbc$ satisfying the following properties  except with negligible probability:  
 
	\begin{itemize}
		\item Agreement and totality. Same to w$\rbc$.
		\item Validity. If no honest node invokes $abandon()$, then an honest sender's input would eventually be output by all honest nodes.
		\item External Validity. An honest node   outputs $v$ s.t. $Q(v)=1$.
		\item Abandonability. If a sender inputs a value $v$ to avw$\rbc$ and remains so-far honest upon the moment when $n-2f$ forever honest nodes have invoked their $abandon()$ interface, then for any (strongly) adaptive adversary $\adv$, it cannot make any honest node to output a value $v'\ne v$.
	\end{itemize}
\end{definition}

\begin{figure}
	\begin{tcolorbox}[left=0mm]
		\vspace{-0.1cm}
		\begin{center}
			{\bf  Protocol	of avw$\rbc$  with  sender $\node_s$}  (for each   $\node_i$)
		\end{center}
		\vspace{-0.4cm}
		\noindent\makebox[\linewidth]{\rule{\textwidth}{0.4pt}}
		
		\begin{footnotesize}
			
			\hspace{0.3cm}{\bf Input:} a value $v_i$ s.t. $Q(v_i)=1$
			
			\hspace{0.3cm}{\bf Initialization:} $ban \leftarrow \textbf{false}$, $val \leftarrow \perp$

			\begin{algorithmic}[1]
				\IF {$\node_i=\node_s$}
				\STATE{multicast $\mathsf{VAL}(v_i)$} 
				\EndIf
				\Upon{receiving $\mathsf{VAL}(v_s)$ from $P_s$}
				\Wait{ until $Q(v)$ holds or $ban$ = $\textbf{true}$}
				\IF{$ban$ = $\textbf{false}$}
				\STATE{$val \leftarrow v_s$}
				\STATE{multicate $\mathsf{ECHO}(\hash(v_s))$}
				\EndIf
				\EndWait 
				\EndUpon
				\Upon{receiving $n$-$f$ matching $\mathsf{ECHO}(h)$ from distinct nodes}
				\STATE{multicast $\mathsf{READY}(h)$}
				\EndUpon
				\Upon{receiving $f+1$ matching $\mathsf{READY}(h)$ from distinct nodes and hasn't sent $\mathsf{READY}$ message}
				\STATE{multicast $\mathsf{READY}(h)$}
				\EndUpon
				\Upon{receiving $n$-$f$ matching $\mathsf{READY}(h)$ from distinct nodes}
				\STATE{$wr\text{-}deliver(h)$}
				\EndUpon
				\Upon{$wr\text{-}deliver(h)$}
				\Wait{ until $\hash(val) = h$}
				\STATE{$r\text{-}deliver(val)$}
				\EndWait
				\EndUpon
				\Upon{the local interface $abandon()$ is invoked}
				\STATE{$ban$ = \textbf{true}}
				\EndUpon
			\end{algorithmic}
		\end{footnotesize}
		 \vspace{-0.15cm}
	\end{tcolorbox}
	\begin{center}
		\vspace{-0.4cm}
		\caption{Our avw$\rbc$ protocol. Code is for each  $\node_i$.}
		\label{algo:rbc}
	\end{center}
	\vspace{-0.9cm}
\end{figure}

As shown in Figure \ref{algo:rbc}, we construct avw$\rbc$ by modifying the w$\rbc$ variant of Bracha's $\rbc$. External validity is trivial to realize by applying validation check after each node receives the  broadcasted value  as in many studies \cite{abraham2021reaching,duan2023practical}. However,    abandonability has to be carefully tailored, since we   need   preserve all other properties of   w$\rbc$ such as weak totality, which means if some   honest node already outputs, any honest shall still output even if it has invoked  $abandon()$. 
We identify a unique place in Bracha's $\rbc$ to take  $abandon()$ into effect with preserving totality, that is to stop sending echo messages after invoking $abandon()$  (lines 4-5 in Figure \ref{algo:rbc}).

\smallskip
\noindent
{\bf FIN-$\mvba$-Q protocol.} Given avw$\rbc$, we now are ready to construct FIN-$\mvba$-Q---an adaptively secure $\mvba$ protocol with 1/2 quality.
Its formal description is  in Figure \ref{fig:finmvba}. Slightly informally, it executes in next three phases:
\begin{enumerate}[label=(\arabic*)]
	\item \textit{Broadcast of input} (lines 1-6). When    receiving  the input $v_i$,  $\node_i$ activates avw$\rbc_i$ to broadcast $v_i$. 
	If   avw$\rbc_j$ $wr\text{-}deliver$ $h_j$,  $\node_i$ multicasts a message $\mathsf{FIN}(j,v_j)$ and records $h_j$ in   $H_i$ (lines 2-4). Once avw$\rbc_j$ $r\text{-}deliver$ $v_j$, $\node_i$ records $v_j$ in   $V_i$ (lines 5-6). 
	\item \textit{Finish all broadcasts with abandon} (lines 7-10). Then, $\node_i$ waits for   $n-f$ $\mathsf{FIN}(j,h)$ messages from distinct nodes  carrying the same $h$, $\node_i$ then records avw$\rbc_j$ as finished (i.e., set $F_{i}\big[j\big]$   true in line 8) and accepts $v_j$ as avw$\rbc_j$ output (also in line 8). Upon $n-f$ avw$\rbc$ instances are labeled as finished, $\node_i$ would invoke the $abandon()$ interface of all avw$\rbc$ instances (line 10).   
	\item \textit{Repeating vote until decide an output} (lines 11-27). Once $\node_i$ abandons in all broadcasts, it loops a voting phase  consisting of three steps: (i)  in line 12, invoke coin flipping to randomly $election$ a node $\node_k$; (ii) in line 13-18, invoke $\raba$ to vote whether decide $h_k$'s pre-image as candidate output or not; (iii) in line 19-28, if $\raba$ output 0, $\node_i$ will enter the next iteration, otherwise, $\node_i$ will wait for the corresponding   $h_k$ and $v_k$ (line 20 and 24)  to output and terminate.
\end{enumerate}

\begin{figure}
	\begin{tcolorbox}[left=0mm]
					\vspace{-0.1cm}
		\begin{center}
			{\bf Protocol of FIN-$\mvba$-Q}  (for each   $\node_i$)
		\end{center}
			\vspace{-0.4cm}
		\noindent\makebox[\linewidth]{\rule{\textwidth}{0.4pt}}
		
		\begin{footnotesize}

		\hspace{0.3cm}{\bf Input:} a value $v_i$ s.t. $Q(v_i)=1$
		
		\hspace{0.3cm}{\bf Initialization:} $r \leftarrow 0$, $H_i \leftarrow \left[\perp \right]^n$, $V_i \leftarrow \left[\perp \right]^n$, $F_i \leftarrow \left[\perp \right]^n$

		\begin{algorithmic}[1]\label{algo:MVBA}
			\STATE{start avw$\rbc_i$ with input $v_i$}
			\If{$h_j$ is $wr\text{-}delivered$ in avw$\rbc_j$}
			\STATE{$H_{i}\big[j\big]$ $\leftarrow$ $h_j$}
			\STATE{multicast $\mathsf{FIN}(j,h_j)$}
			\EndIf
			\If{$v_j$ is $r\text{-}delivered$ in avw$\rbc_j$}
			\STATE{$V_{i}\big[j\big]$ $\leftarrow$ $v_j$}
			\EndIf
			\Upon{receiving $n$-$f$ same $\mathsf{FIN}(j,h_j)$ from distinct nodes}
			\STATE{$F_{i}\big[j\big]$ $\leftarrow$ $\textbf{true}$; $H_{i}\big[j\big]$ $\leftarrow$ $h_j$}
			\EndUpon
			\Wait{ until $n-f$ values in $F_{i}$ are $\textbf{true}$}\label{line:enoughfin} 
			\STATE{avw$\rbc_j$.$abandon()$ for every $j \in [n]$}\label{line:abandon}
			\REPEAT {\textcolor{orange}{~~~// repeatedly attempt to output a finished  broadcast}} 
			\STATE{$k \leftarrow \mathsf{Election}(r)$ \textcolor{orange}{~~~// coin flipping to randomly select a node}}
			\If{$H_{i}\big[k\big]  \neq  \perp$}
			\STATE{$raba$-$propose$ 1 for $\raba_r$}
			\Else
			\STATE{$raba$-$propose$ 0 for $\raba_r$}
			\If{later $H_{i}\big[k\big]  \neq  \perp$}
			\STATE{$raba$-$repropose$ 1 for $\raba_r$}
			\EndIf
			\EndIf
			\If{$\raba_r$ output $1$}
			\Wait{ until $H_{i}\big[k\big] \neq  \perp$}
			\If{$V_{i}\big[k\big] \neq  \perp$}
			\STATE{multicast $\mathsf{Value}(V_{i}\big[k\big])$}
			\Else
			\Wait{ $\mathsf{Value}(v)$ such that $\hash(v)=H_{i}\big[k\big]$}
			\EndWait
			\EndIf
			\STATE{$mvba\text{-}decide(V_{i}\big[k\big]$) and $halt$}
			\EndWait
			\EndIf
			\STATE{$r \leftarrow r+1$}
			\UNTIL{$halt$}
			\EndWait
		\end{algorithmic}
	    \end{footnotesize}
		 \vspace{-0.15cm}
	\end{tcolorbox}
	
	\begin{center}
		\vspace{-0.4cm}
		\caption{Our FIN-$\mvba$-Q   protocol. Code is for each  $\node_i$.}	
		\label{fig:finmvba}
	\end{center}
	\vspace{-1cm}
			
\end{figure}

The above design ensures   all broadcasts are   ``abandoned'' while the adversary learns  coin flipping, thus preventing an adaptive rushing adversary to manipulate $\mvba$ output. More formally, the security of FIN-$\mvba$-Q can be summarized as:
 
\begin{theorem}
	FIN-$\mvba$-Q securely realizes   $\mvba$    with $1/2$ quality except with negligible probability.
\end{theorem}

\noindent
{\bf Security analysis.} Security of FIN-$\mvba$-Q is also intuitive:
\begin{itemize}
	\item For agreement, this is   because avw$\rbc$, $\raba$ and leader election are all agreed. So any two nodes' output must be same (as it must be a value that is broadcasted by the same avw$\rbc$).
	\item For termination, the key lemma is: at the moment when the $(f+1)$-th honest node invokes the leader election (and also abandons all avw$\rbc$s), there must be at least $2f+1$ distinct nodes' avw$\rbc$s have delivered to at least $f+1$ honest nodes. Recalling that for these $2f+1$ avw$\rbc$s, their corresponding $\raba$s  must return 1 (if were executed), so for in each iteration, there is at least $2/3$ probability to draw an avw$\rbc$ whose corresponding $\raba$ return 1.
	Hence after $k$ iterations, the chance of not terminating becomes a negligible probability $(1/3)^k$, which also implies the expected constant round complexity.
	\item For quality, we can  make a corollary of termination proof: for each iteration, there is at least  $1/3$ probability to decide an honest avw$\rbc$ as output, at most $1/3$ probability to draw an malicious node's avw$\rbc$ as output, and at most $1/3$ probability to elect some unfinished avw$\rbc$ and thus enter the next iteration. So the lower bound of quality is $\sum_{r=0}(1/3)(1/3)^r=1/2$.
	\item For external validity, this is true because: (i) the external validity of avw$\rbc$ prevents the honest nodes from receiving invalid input (or digest of invalid input); (ii) if an avw$\rbc$ does not deliver output to any honest node, its corresponding $\raba$ wouldn't return 1 as no honest node   proposes 1, so it is never decided as output.
\end{itemize}

We defer the full proof to Appendix \ref{app:mvba-proof} in  Supplementary.

 \smallskip
 \noindent
 {\bf Complexity analysis}. FIN-$\mvba$-Q has expected constant round complexity and $\bigO(n^3)$ message complexity.
 Its communication cost is $\bigO(|m|n^2 + \lambda n^3)$ where $|m|$ represents $\mvba$'s input length and $\lambda$ represents the length of hash digest in avw$\rbc$.
Remark that  the communication cost of FIN-$\mvba$-Q is as   same as FIN-$\mvba$ \cite{duan2023practical}. Though such $\bigO(|m|n^2 + \lambda n^3)$ communication cost is best so far in the signature-free setting, it is more costly than   the
state-of-the-art signature-involved $\mvba$ protocol \cite{lu2020dumbo} that attains expected $\bigO(|m|n + \lambda n^2)$ communication cost.

	\vspace{-0.2cm}
\subsection{Whether is $\finng$   sufficiently   scalable or NOT?}

Careful readers might realize that 
the design of $\finng$ can be thought of as a signature-free version of $\dumbong$,
involving a particular trade-off to reduce  authenticator complexity  
at a price of other overheads like cubic message complexity.
Observing that, a natural question might raise: 
is such stringent signature-free approach
the best practice in real-world global deployment setting  with affordable bandwidth (like 100 Mbps or less) and several hundreds of nodes?

To better understand this natural question, we test two state-of-the-art $\mvba$ protocols---
FIN-$\mvba$ (signature-free but cubic messages) and  $\dumbo$-$\mvba$  \cite{lu2020dumbo} (involving signatures but only quadratic messages) to estimate the tendency of their latency degradation in bandwidth-limited settings with $n$=196 nodes.
For fair comparison, we take a vector of $n$ integers as input to FIN-$\mvba$,
and use $n$ QCs (which is much larger) as input to   $\dumbo$-$\mvba$.
As Figure \ref{fig:mvba-latency} illustrates, our  experimental comparison between FIN-$\mvba$ and  $\dumbo$-$\mvba$ reveals: (i) FIN-$\mvba$ and $\dumbo$-$\mvba$ have comparable latency only if each node has more than 100 Mbps bandwidth; (ii) when available bandwidth is less than 50 Mbps, FIN-$\mvba$'s latency would be dramatically increased.

In short, this evaluation   affirms our conjecture that the stringent signature-free approach of $\finng$ 
might be not the best practice in  the  deployment environment with restricted bandwidth, motivating us further explore even more scalable asynchronous BFT protocols in the signature-involved setting. 


\begin{figure}[H] 
	\vspace{-0.2cm}
	\begin{center}
		\includegraphics[width=7.7cm]{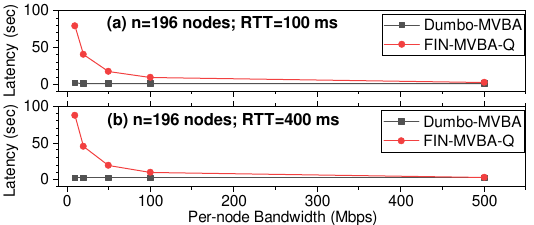}\\
		\vspace{-0.25cm}
		\caption{Latency of signature-free $\mvba$ (i.e. FIN-$\mvba$-Q) and signature-involved $\mvba$ (i.e. $\dumbo$-$\mvba$) for $n=196$ nodes.}	
		\label{fig:mvba-latency}
	\end{center}
	\vspace{-0.3cm}
\end{figure}

\section{Optimize   Quorum Certificate   to Best Practice}
\label{sec:qc}

As our initial attempt $\finng$   might not scale very well in the important setting of restricted bandwidth,
we continue our concentration on how to scale up   asynchronous BFT consensus using   signatures.
To begin with, this section focuses on  optimizing   
 QC  implementations, as its prevalent usage has led up to a primary scalability bottleneck   in existing designs. 

\smallskip
\noindent
{\bf Batch verification with   blocklist: the right way of implementing QCs}.
We take       consistent broadcast (CBC)   \cite{reiter1994secure} as the simplistic context to estimate the CPU cost of different QC implementations from different signatures.

As Figure \ref{fig:cbc-example} (a) plots, CBC has a designated sender that waits for $n-f$ valid signatures from distinct nodes, such that it can aggregate/concatenate these $n-f$ valid signatures to form a QC and spread it out. 
We denote this approach ``individual verification'' since the aggregation occurs after validating all $n-f$ individual signatures. 
Figure  \ref{fig:cbc-example} (b) presents a folklore alternative \cite{opticoin,optiagg} that directly aggregates (unverified) signatures and then performs batch verification.
Intuitively, the ``batch verification'' approach  is   promising to take the best advantage of aggregate signature, because it avoids verifying   individual signatures in normal case. 

We then conduct a comprehensive benchmarking of QCs built from ``individual verification'' and ``batch verification'', with different types of signatures (including ECDSA, Schnorr signature, BLS signature, and BLS threshold signature).
Note that ECDSA signature cannot be non-interactively aggregated, and   QC from it  simply concatenates $n-f$ valid signatures. 
Schnorr signature can be non-interactively half-aggregated \cite{chen2022half,chalkias2021non}, so QC from it has a smaller size that is about half of $n-f$ concatenated signatures, and more importantly, its batch verification can be faster than trivially verifying $n-f$ concatenated Schnorr signatures.\footnote{Similar to Schnorr signature, EdDSA can   be half-aggregated and might have very tiny  speed-up due to  Edward curve. However, their   performances are essentially comparable. Thus we didn't do   redundant tests of EdDSA.} BLS signature can be fully aggregated. 
In order to push efficiency, we use a straight way of computing BLS multi-signature by taking summation of all received signatures (under the standard KOSK model),
so   aggregation   is much faster than the threshold version of BLS, as the latter one's aggregation needs costly interpolation in the exponent \cite{tomescu2020towards}. 
Our evaluations  consider  similar security levels for different signatures. ECDSA and Schnorr are implemented over secp256k1 curve with about 128-bit security \cite{bitcoin-core}. For BLS signature, we work on a notable and widely-adopted pairing-friendly elliptic curve BLS12-381 \cite{barreto2003constructing,bls}, which    was designed for 128-bit security and now is still considered with  117-120 bit security despite   recent cryptanalysis \cite{kim2016extended}.

\begin{figure}[H] 
	\vspace{-0.2cm}
	\begin{flushleft}
		\includegraphics[width=9.2cm]{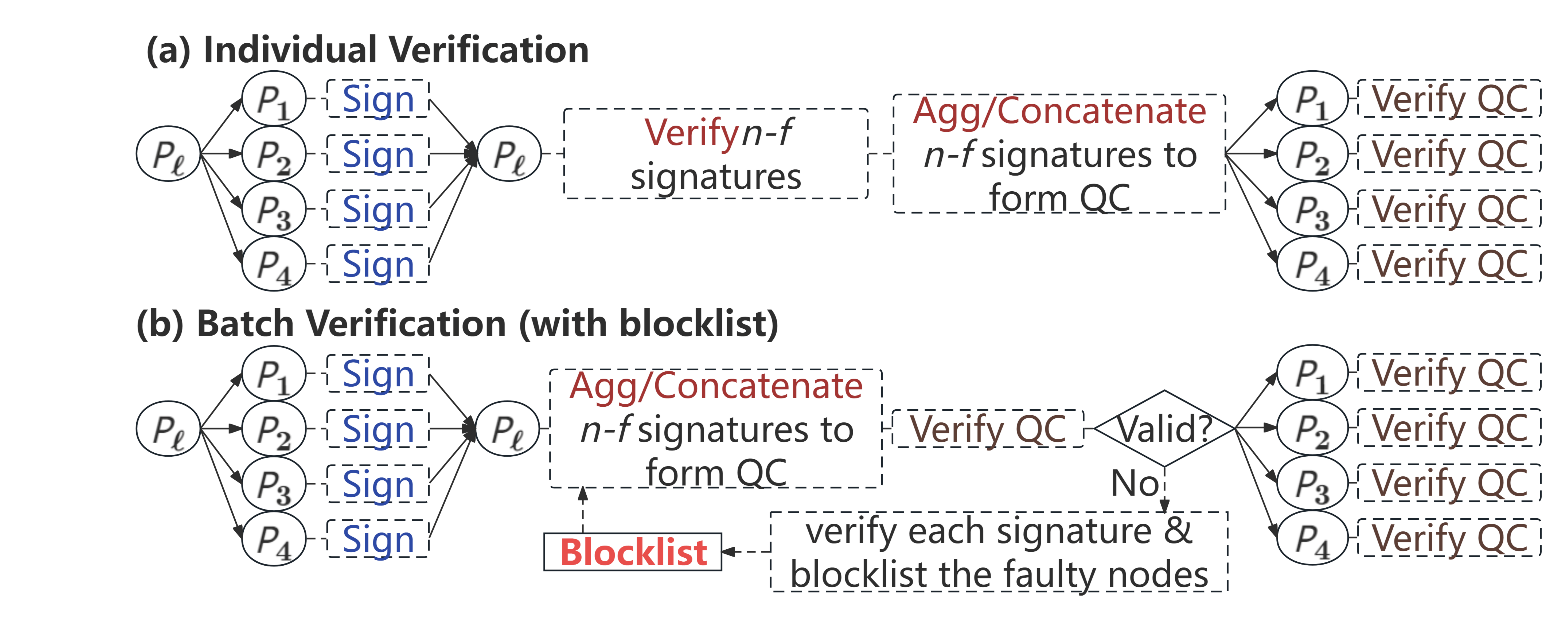}\\
		\vspace{-0.15cm}
		\caption{QC  generation in CBC: (a)   verify individual signatures  then aggregate; (b) aggregate then verify a batch.}	
		\label{fig:cbc-example}
	\end{flushleft}
	\vspace{-0.3cm}
\end{figure}

\noindent
\underline{\smash{\emph{BLS multi-signature with  batch verification  would outperform}}}. 
The detailed evaluation    of CPU latency in a CBC protocol is plotted in Figure \ref{fig:signature}.
Evidently,  QC implementation from BLS multi-signature with using batch verification ({\color{violet} the violet line} and {\color{gray} the gray line})   vastly outperforms all other instantiations in terms of computation cost. Note   the minor difference between the violet and gray lines reflects the choice of implementing signatures in $\mathbb{G}_1$   or    $\mathbb{G}_2$ groups   over BLS12-381 curve.

\noindent
\underline{\smash{\emph{Maintain  blocklist to prevent performance degradation attack}}}. 
However, the approach of batch verification  might face serous performance downgrade under attacks, particularly     when malicious nodes   send false signatures. This issue can be   addressed by maintaining a blocklist containing  the identities of those malicious nodes. 
As Figure \ref{fig:cbc-example} (b) shows, when the aggregated multi-signature fails in batch verification, one  verifies all individual signatures and add the malicious nodes   sending fake signatures to a blocklist. Later, the blocklist can be used to exclude the  faulty nodes from   aggregation. 

\begin{figure}[H] 
	\vspace{-0.2cm}
	\begin{center}
		\includegraphics[width=7.5cm]{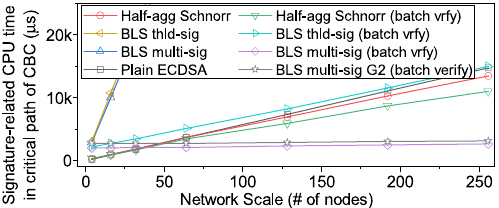}\\
		\vspace{-0.2cm}
		\caption{CPU latency (related to digital signatures) in the critical path of one-shot CBC on varying system scales.}	
		\label{fig:signature}
	\end{center}
	\vspace{-0.3cm}
\end{figure}



\noindent
{\bf Further aggregation of multiple QCs:  efficiently transfer a vector of QCs on different messages}.
While implementing QC from BLS multi-signature, 
one more benefit largely overlooked by the distributed computing community is that many such QCs on different messages can be further aggregated.
This bonus property becomes particularly intriguing in asynchronous BFT protocols like $\tusk$ and $\dumbong$ because their scalability bottleneck   stems from   multicasting $\bigO(n)$ QCs.
As illustrated in Fig  \ref{fig:agg-qc}, 
while multicasting $n$ QCs made from BLS multi-signature, 
we can leverage the aggregation property to combine all signature parts in these QCs and reduce their size to less than half of the original (cf. Fig \ref{fig:agg-qc}).

\begin{figure}[H] 
	\vspace{-0.2cm}
	\begin{flushleft}
		\includegraphics[width=8.5cm]{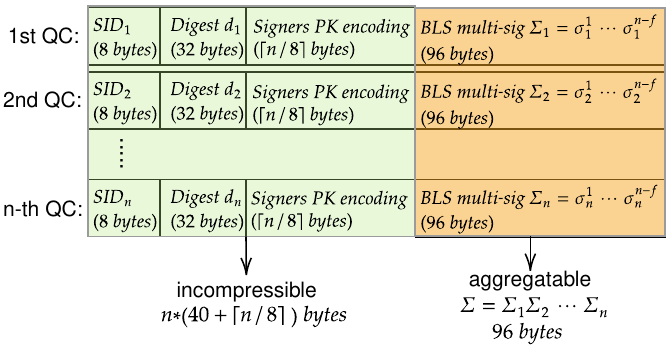}\\
		\vspace{-0.3cm}
		\caption{Aggregating many QCs made from BLS multi-sig.}	
		\label{fig:agg-qc}
	\end{flushleft}
	\vspace{-0.3cm}
\end{figure}

\smallskip
\noindent
{{\bf{More tips on efficient QC implementation}}}. 
Besides     various aggregation   techniques, we remark a few other tips of implementing QCs from BLS signatures in real-world BFT systems:

\begin{itemize}
	\item  \underline{\smash{\emph{BLS signature  in $\mathbb{G}_1$ group  is more favorable than  $\mathbb{G}_2$}}}. BLS12-381 curve is equipped with asymmetric paring $e: \mathbb{G}_1 \times \mathbb{G}_2 \rightarrow \mathbb{G}_T$.
	When implementing BLS signature over the curve, one can let signature in $\mathbb{G}_1$ and public key  in $\mathbb{G}_2$, or vice versa.
	We recommend     signature   in $\mathbb{G}_1$, because (i) $\mathbb{G}_1$ group element has shorter size, and (ii) there is   saving in CPU  time for $n\le 256$ (Figure \ref{fig:signature}).
	\item \underline{\smash{\emph{Encode public keys relating a   multi-signature  by $n$ bits}}}. Moreover, in most realistic proof-of-stake and consortium blockchains, participating nodes' public keys can be mutually known to each other.  So the identities of  signing nodes in a multi-signature can   be encoded by   $n$ bits instead of transferring their public keys. Here each bit represents whether a   node has signed (1) or not (0). 
\end{itemize}

\section{$\jumbo$: Making signature-involved asynchronous BFT Truely Scalable}
\label{sec:jumbo-impl}


Given    concrete optimizations of QCs,
it is  worth noting that $\tusk$ and $\dumbong$ still suffer from cubic authenticator overhead per decision,
because they  still let each node to multicast a vector of $\bigO(n)$ QCs. Even if using our technique in Section \ref{sec:qc}, this cost is still a cubic in total, becuase the aggregate $\bigO(n)$ BLS-based QCs  carry $\lambda n + n^2$ bits and need $\bigO(n)$ cryptographic pairings to verify.
Facing   still   asymptotically large authenticator overhead, 
this section   distills a scalable instantiation of $\dumbong$---$\jumbo$, 
to  asymptotically reduce  the   authenticator complexity by introducing  information dispersal.
$\jumbo$ can reduce the number of multicasting $n$ QCs from $\bigO(n)$ to $\bigO(1)$, and
in combination with our optimized QC implementation,  it can seamlessly scale up with   hundreds of nodes in   bandwidth-starved   settings.

\smallskip
\noindent
{\bf Revisiting the core technique of  provable dispersal}. 
Before diving into  the  details of our scalable construction, we first  briefly recall   our key technical component---{asynchronous provable dispersal broadcast ($\mathsf{APDB}$) \cite{lu2020dumbo,gao2022dumbo}}.
Informally,  $\mathsf{APDB}$   allows a designed sender to disperse a large input like $n$ QCs to the whole network at an overall communication cost of merely $O(1)$ QCs.
Namely, the communication cost of dispersing $n$ QCs can be   $O(1/n)$ of multicasting them.
Moreover, the sender can generate a proof at the end of dispersal, 
 such that a valid proof attests that the dispersed $n$ QCs can later  be consistently recovered by the whole network.


More precisely,   $\mathsf{APDB}$   can be formally defined as:

\begin{definition}
	Syntactically, an
	$\mathsf{APDB}$ consists of    two   subprotocols ($\PD$, $\RC$)  with a   validating function $\ValidateLock$: 
	\begin{itemize}
		\item {\em $\PD$ subprotocol}. 
		In the $\PD$ subprotocol among $n$ parties,  a designated sender $\node_{s}$ inputs a value $v\in\{0,1\}^\ell$, 
		and aims to split $v$ into $n$ code fragments and disperses each $j$-th fragment to the corresponding node $\node_j$. 
		If a party terminates in $\PD$, it shall output a string $\store$, and if  the node is sender, it shall additionally output a  $\lock$ string. 
		
		\item {\em $\RC$ subprotocol}.  
		In the $\RC$ subprotocol among $n$ parties, all honest nodes take   $\store$ and $\lock$ strings outputted by  $\PD$ subprotocol as   input, and aim to output a value $v' \in \{0,1\}^\ell \cup \bot$.
		
		\item {\em $\ValidateLock$ function}. It takes a $\lock$ string as input and returns 0 (reject) or 1 (accept).
	\end{itemize}
	
	An $\mathsf{APDB}$ protocol ($\PD$, $\RC$)  shall satisfy the following properties except with negligible probability:
	
	\begin{itemize}	
		\item {\em Termination}. 
		If the sender $\node_s$ is honest and all honest nodes activate the $\PD$ protocol, 
		then each honest node would output $\store$ in $\PD$;
		additionally,  $\node_s$ also outputs valid $\lock$ in $\PD$ s.t. $\ValidateLock(lock)$=$1$. 
		
		\item {\em Recast-ability}. 
		If all honest parties invoke $\RC$ with inputting the output of $\PD$ 
		and at least one honest party inputs a valid $\lock$  s.t. $\ValidateLock(lock)$=$1$, then:
		(i) all honest parties can eventually recover a common value $v' \in \{0,1\}^\ell \cup \bot$; 
		(ii) if the sender is  honest and disperses  $v$ in $\PD$, then all honest parties can  recover $v$ in $\RC$.
		
	\end{itemize}
	
\end{definition}

An $\mathsf{APDB}$ protocol   satisfying the above definition can be efficiently  implemented as in  Figures \ref{algo:dispersal} and \ref{algo:recast}, through  simplifying the design from \cite{lu2020dumbo}. 
  $\PD$ subprotocol   has a very simple two-round structure: sender  $\xrightarrow[]{\scriptscriptstyle \PreDisp}$ parties  $\xrightarrow[]{\scriptscriptstyle \Store}$ sender.
In this way, the    communication cost of using $\PD$ to disperse $n$ broadcast QCs can be brought down to minimum, which asymptotically is  only  $O(1/n)$ of multicasting $n$ QCs.

\begin{figure}[H]
	\vspace{-0.3cm}
	\begin{center}
		\includegraphics[width=7.8cm]{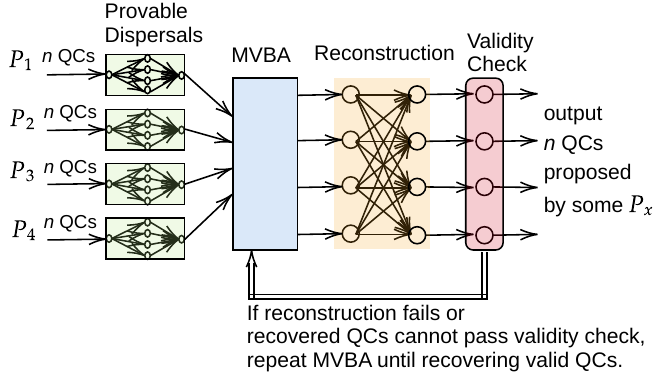}\\
		\vspace{-0.25cm}
		\caption{Dispersal-then-recast paradigm of  $\dumbo$-$\mvba$     \cite{lu2020dumbo}.}	
		\label{fig:dumbo-mvba}
	\end{center}
	\vspace{-0.4cm}
\end{figure}

\begin{figure}
	\begin{tcolorbox}[left=0mm]
		\vspace{-0.1cm}
		\begin{center}
			{\bf Protocol of Provable Dispersal ($\PD$)   with   sender $\node_{s}$}
		\end{center}
		\vspace{-0.4cm}
		\noindent\makebox[\linewidth]{\rule{\textwidth}{0.4pt}}
		
		\begin{footnotesize}

			
			\hspace{0.3cm}{\bf Initialization:} $S \leftarrow \{\ \}$

			\begin{algorithmic}[1]\label{algo:MVBA}

				\vspace*{0.1 cm}
				\Statex  /* Protocol for {\color{black} the sender $\node_s$} */
				\vspace*{0.1 cm}
				
				\Upon {receiving an input value $v$}  
				
				\State encode $v$ by erasure code to get the  code word $\bf{m}$ of $n$    fragments
				\State use   $n$    fragments of $\bf{m}$ as leaves to compute   Merkle tree root $mt$ 
				\FOR {each $j \in [n]$}               
				\State \textbf{let} $\store := \langle mr, m_j,\pi_{j} \rangle$ where $\pi_{j}$ is the Merkle  tree path proving that $m_j$ is the $j$-th leaf committed to $mr$
				\State \textbf{send} $(\PreDisp, \DID, \store)$ to $\node_{j}$    
				\ENDFOR
				
				\Wait { until {$|S|=2f+1$}}    {\textcolor{orange}{~~// Aggregate threshold signature}}      
				\State $\Sigma \leftarrow$ $\Agg(S)$ and \textbf{let} $\lock := \langle mr, \Sigma \rangle$
				\State    \textbf{output} $\lock$ and terminate
				\EndWait

				\EndUpon		
				
				\vspace*{0.1 cm}
				\Upon {receiving valid $(\Store, \DID, \sigma_{j})$ from $\node_j$}           
				\State $S\leftarrow S\cup (j, \sigma_{j})$ {\textcolor{orange}{~~// For a neat  description, here is not batch verification, but we can implement batch verification in practice.}}
				\EndUpon

				\vspace*{0.2 cm}
				\Statex  /* Protocol for {\color{black} each party $\node_i$} */
				\vspace*{0.1 cm}
				
				\Upon {receiving   $(\PreDisp, \DID, \store)$ from sender $\node_s$}
				\If {$\store=\langle mr,   m_i, \pi_{i}\rangle$ carries a valid $\pi_{i}$ proving that $m_i$ is the $i$-th leaf committed to Merkle tree root $mr$}    
				\State $\sigma_{i} \leftarrow \Sign(sk_i, \langle  \DID, mr \rangle)$; \textbf{send} $(\Store, \DID, \sigma_{i})$ to $\node_s$
				\State \textbf{output} $\store$  and terminate if $\node_i\ne \node_s$
				\EndIf
				\EndUpon

				

			\end{algorithmic}
		\end{footnotesize}
		 \vspace{-0.15cm}
	\end{tcolorbox}
	
	\begin{center}
		\vspace{-0.4cm}
		\caption{Provable dispersal   protocol adapted from \cite{lu2020dumbo}. Code is for each  $\node_i$.}	
		\label{algo:dispersal}
	\end{center}
	\vspace{-0.8cm}
	
\end{figure}

\begin{figure}
	\begin{tcolorbox}[left=0mm]
		\vspace{-0.1cm}
		\begin{center}
			{\bf Protocol of Reconstruction ($\RC$)} (for each $\node_{i}$)
		\end{center}
		\vspace{-0.4cm}
		\noindent\makebox[\linewidth]{\rule{\textwidth}{0.4pt}}
		
		\begin{footnotesize}

			
			\hspace{0.3cm}{\bf Initialization:}   $C \leftarrow [\ ]$ (i.e, a dictionary keyed by Merkle tree root)

			\begin{algorithmic}[1]

				\vspace*{0.15 cm}
				\Upon {receiving input $(\store, \lock)$}  
				\If{$\lock \ne \emptyset$ }	   
				\State \textbf{multicast} $(\RecastCheck, \DID, \lock)$ to all
				\EndIf
				\If{$\store \ne \emptyset$}  
				\State \textbf{multicast} $(\RecastVal, \DID, \store)$ to all
				\EndIf
				\EndUpon
				\Upon { receiving $(\RecastCheck, \DID, \lock)$}
				\If{ $\ValidateLock(\DID, \lock) = 1$}       
				\State\textbf{multicast} $(\RecastCheck, \DID, \lock)$ to all, if  was not sent before
				\State parse $\lock$ as   $\langle mr, \Sigma \rangle$
				\Wait { until $|C[mr]| = f+1$}
				\State decode $C[mr]$ by erasure code to obtain $v$      
				\State encode $v$ by erasure code to get code word $\bf{m}$ 
				\State use $\bf{m}$ as leaves to compute Merkle tree root $mr'$        
				\If{$mr=mr'$}					 
				\textbf{return} $v$           
				\Else  
				{\ }\textbf{return} $\bot$
				\EndIf
				\EndWait	
				\EndIf
				\EndUpon
				\Upon {receiving $(\RecastVal, \DID, \store)$ from $\node_{j}$}
				\If {$\store=\langle mt,   m_j, \pi_{j}\rangle$ carries a valid $\pi_{j}$ proving that $m_j$ is the $j$-th leaf committed to Merkle tree root $mr$}    
				\State $C[mr] \leftarrow C[mr] \cup (j, m_j)$ 
				\EndIf
				\EndUpon	
				
				\noindent\makebox[\linewidth]{\rule{\textwidth}{0.4pt}}
				
				\vspace*{0.1 cm}
				\Statex  /* Public $\ValidateLock$ function */
				\vspace*{0.2 cm}
				
				\Function{\textbf{function} $\ValidateLock$}{$(\lock)$:}
				\State parse $\lock$ as $\langle mr, \Sigma \rangle$
				\State verify   $\Sigma$ is a valid $(2f+1)$ threshold signature on   $\langle \DID, mr \rangle$
				\State return 1 (accept) if     verification passes, otherwise return 0 (reject)
				\EndFunction
				
			\end{algorithmic}
		\end{footnotesize}
		\vspace{-0.15cm}
	\end{tcolorbox}
	
	\begin{center}
		\vspace{-0.4cm}
		\caption{Provable dispersal   protocol adapted from \cite{lu2020dumbo}. Code is for each  $\node_i$.}	
		\label{algo:recast}
	\end{center}
	\vspace{-1cm}
	
\end{figure}

\smallskip
\noindent
{\bf Inspiration of $\jumbo$: $\dumbong$ is    friendly  for   reducing the  authenticator overhead  by using $\mathsf{APDB}$!} 
It comes to our  attention that $\dumbong$ has a very unique place leading up to its cubic authenticator complexity.
That is every node using a vector of $\bigO(n)$ QCs as $\mvba$ input.
The neat structure of $\dumbong$   hints at introducing $\mathsf{APDB}$   to extend its $\mvba$ component,  thus asymptotically reducing overall communication cost for large inputs like $n$ QCs.

  The communication-efficient $\mvba$   protocol  extended by  $\mathsf{APDB}$ also has a nickname $\dumbo$-$\mvba$ \cite{lu2020dumbo}. As shown in Figure \ref{fig:dumbo-mvba}, it proceeds as:  every node avoids   directly  multicasting its large input, but employs the more efficient provable dispersal  protocol  to spread much shorter encode  fragments instead; then, if some node finishes its dispersal, it can generate a dispersal QC attesting  its successful  spread of  $n$ broadcast QCs; every honest node can finish its own dispersal and thus uses its dispersal QC (which is a single QC instead of $n$ QCs) to invoke the original $\mvba$ protocol;  finally, all nodes repeat     $\mvba$ protocols, until some $\mvba$ picks a dispersal QC that allows them collectively recover  a vector of $n$ valid broadcast QCs as their      $\mvba$  result.


Clearly, after applying  the dispersal technique as above,  we can reduce the   cubic authenticator bottleneck in $\dumbong$   to only quadratic.  This is because: (i) $n$ dispersals of $\bigO(n)$ QCs have a communication cost similar to a single multicast of $\bigO(n)$ QCs; (ii) we only have to reconstruct a small constant number of dispersals, which is also asymptotically same to the communication cost of a single multicast of $\bigO(n)$ QCs.

\subsection{$\jumbo$ Protocol}


As depicted by Figure \ref{fig:dumbo-ng},
$\jumbo$ inherits $\dumbong$  to have a couple of concurrent  processes---transaction broadcast and block agreement---with  using provable dispersal to resolve the authenticator bottleneck.
For completeness of presentation, we hereunder briefly  describe the two processes of $\jumbo$, and then introduce its easy-to-implement fairness patch.

  \smallskip
 \noindent
	  {\bf Transaction broadcast  process}. Each node $\node_i$ starts $n$ broadcast threads. Each broadcast   proceeds in slot $s\in\{1,2,3\cdots\}$. $\node_i$ is designated as sender only in the $i$-th broadcast. Each broadcast  proceeds by an increasing-only slot number $s$. Once the sender $\node_i$ enters a slot $s$, it selects a batch of $|B|$ transactions (denoted $tx_s$) from its input buffer,  multicasts $tx_s$  with the current slot index $s$ and  a quorum certificate $QC_{s-1}$  attesting that a transaction batch $tx_{s-1}$ was broadcasted in the preceding slot $s-1$,  then it waits for that $n-f$ distinct nodes return valid signatures on $i||s||\hash(tx_s)$ (where $\hash$ is a cryptographic hash function), and assembles these signatures to form a   quorum certificate $QC_{s}$ and then steps into the next slot $s+1$. On the other side, when any node $\node_j$ (probably non-leader) enters a slot $s$, it waits for receiving the message carrying $tx_s$ and $QC_{s-1}$ from the sender node $\node_i$, checks $QC_{s-1}$  a valid QC and consistent to $tx_{s-1}$ that it received,   returns  a  signature on $i||s||\hash(tx_s)$ to the sender, and proceeds to the next slot $s+1$.
	
	Note that we refrain from reintroducing the pull mechanism  that executes  in companion with broadcast process to help the fetch of missing transactions, because we have described two instantiations of it in Section \ref{sec:jumbo} for $\finng$. 

 \smallskip
\noindent
	  {\bf Block agreement process}. Concurrent to   transaction broadcast, every node starts an agreement process that  executes a sequence of $\mvba$ protocols in consecutive epochs $e\in\{1,2,3,\cdots\}$. Note that  all $\mvba$s are $\dumbo$-$\mvba$ protocols that are extended by using $\mathsf{APDB}$.
	To prepare the input of the current epoch's $\dumbo$-$\mvba$ protocol, each   $\node_i$ locally   maintains a vector of $[\height_1,\cdots,\height_n]$, where each $\height_j$ carries
	  sender  $\node_j$'s latest  broadcast  QC that is  indeed received by $\node_i$.
	Meanwhile,	$\node_i$   maintains another vector of $[\ordered_1,\cdots,\ordered_n]$ to track the $\dumbo$-$\mvba$ output of last epoch, where each $\ordered_j$ represents sender $\node_j$'s highest broadcast QC that was already solicited into consensus output.
	In each epoch, every $\node_i$ waits for that at  least $n-f$ QCs in  $[\height_1,\cdots,\height_n]$ contain slot numbers higher than
	their corresponding items in $[\ordered_1,\cdots,\ordered_n]$,  i.e., $n-f$ broadcasts have indeed progressed since the decision of last block. 
	Then, $\node_i$ takes a snapshot of $[\height_1,\cdots,\height_n]$   as input to  the current $\dumbo$-$\mvba$ protocol.
	To enforce the input validity, $\dumbo$-$\mvba$'s  validity  predicate    is set as: 
	\begin{enumerate}
		\item the input proposal indeed carries $n$ valid QCs associated with $n$ different broadcasts, if passing the check, go to next clause, otherwise, return false (0);
		\item at least $n-f$ QCs have  slot numbers that indeed increase in relative to the last $\dumbo$-$\mvba$  output, if passing the check, go to next clause, otherwise, return false (0);
		\item other non-increasing QCs  shall not carry      slot numbers smaller than the last $\dumbo$-$\mvba$ output,   if   the check passes, return true (1), otherwise, return false (0);
	\end{enumerate}
	
	After invoking the current epoch's $\dumbo$-$\mvba$, $\node_i$ waits for its output, which   are $n$ valid broadcast QCs.
	According to the $\dumbo$-$\mvba$ result, all honest nodes can decide a common new block that solicits the     transactions 
	due to the difference between the current $\dumbo$-$\mvba$ output and the  previous $\dumbo$-$\mvba$ output.	
	Finally, every node's agreement process updates   $[\ordered_1,\cdots,\ordered_n]$ by the newest $\dumbo$-$\mvba$ output and goes to the next epoch.

\smallskip
\noindent
{\bf Enhance   fairness    by an easy-to-implement patch}. 
Although $\jumbo$ could  scale up after applying $\dumbo$-$\mvba$ extension protocol to asymptotically reduce  authenticator overhead,
it   lacks fairness in some extremely adversarial cases, because   its transaction dissemination process allows 
each node $\node_i$ to continually participate in all broadcast threads, thus leaving the adversary a chance to diffuse tremendous transactions ridiculously fast through the malicious nodes' broadcast instances.
Nevertheless,  the desired fairness    could be easily  patched by applying ``speed limit'' in broadcasts.
The high-level idea   is very intuitive: let the honest nodes temporarily halt voting in these   broadcasts that are running too fast in relative to    other slower broadcasts, until the slower broadcast threads catch up. 

To implement the ``speed limit''  idea, a node $\node_i$ maintains a variable $\delta_j$ for each broadcast thread.
The variable $\delta_j$ tracks that the $j$-th broadcast thread has disseminated how many transactions since the height included by last block (from the local view of node $\node_i$), i.e., $\delta_j = \height_j.slot - \ordered_j.slot$, where $\height_j$ tracks the slot number of $\node_j$'s latest broadcast received by $\node_i$ with a valid QC, and $\ordered_j$ tracks last broadcast slot of $\ordered_j$ that has already been decided as output  due to last $\mvba$ result.
Let $\delta$    to be the $(f+1)$-th smallest one of all $\{\delta_j\}_{j\in [n]}$; if there is any $\delta_j$ such that $ \beta \cdot \delta_j \ge  \delta$ and $\delta_j > 0$, $\node_i$ halts   voting in the $j$-th broadcast instance until $ \beta \cdot \delta_j <  \delta$, where $0<\beta<1$ is a fairness parameter (more ``fair'' if closer to 1, because of enforcing the fastest $2f+1$ broadcasts to progress at more similar speeds).

To enforce the above ``speed limit'',  we also need to add a corresponding validity rule to the $\mvba$'s predicate function: 
per each node $\node_i$'s proposal, the most progressed broadcast since last block and the $(f+1)$-th least progressed broadcast since last block  shall have a progress ratio bounded by $1/\beta$.

 \begin{figure}[H]
 	\vspace{-0.2cm}
 	\begin{center}
 		\includegraphics[width=9.1cm]{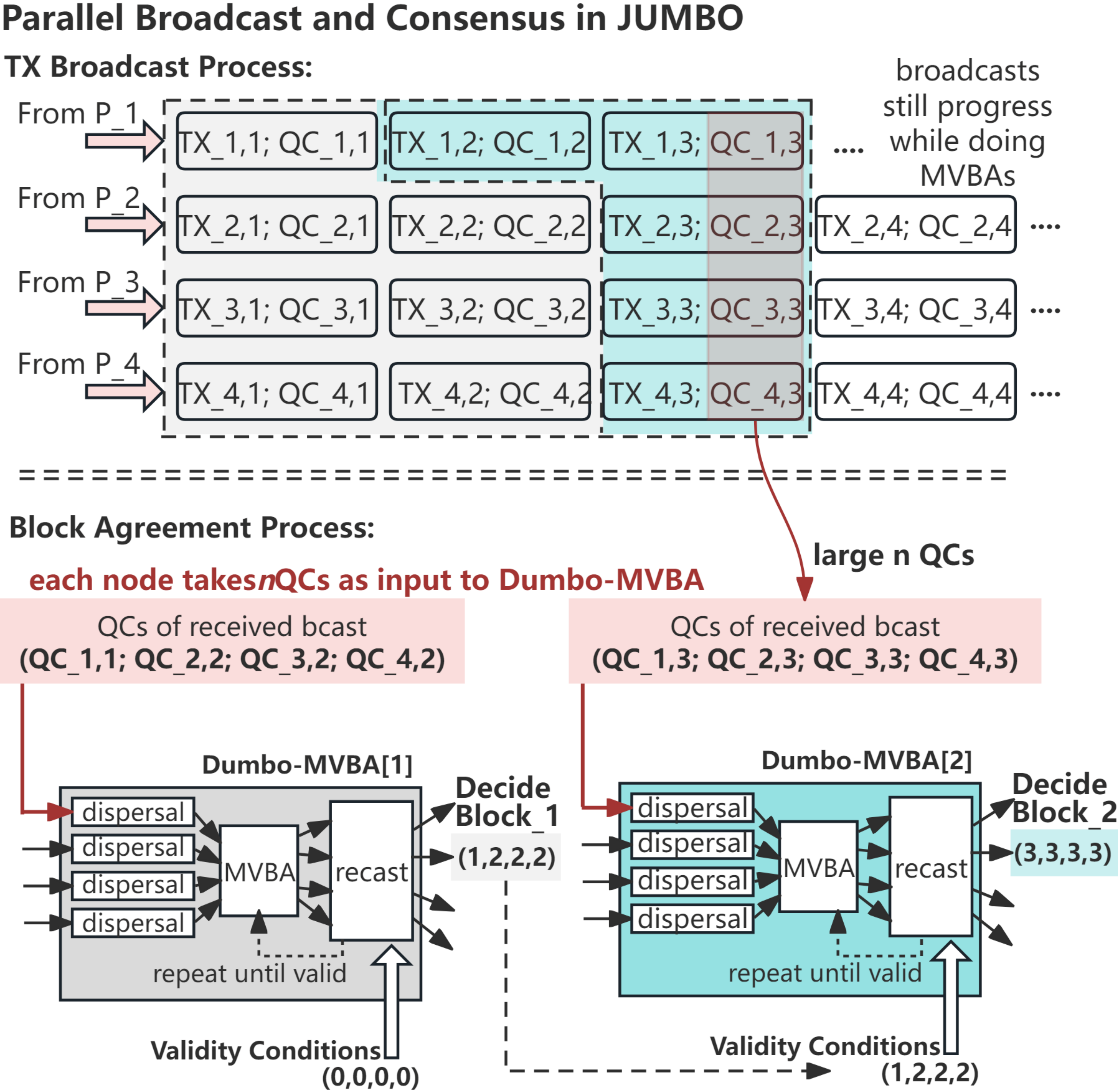}\\
 		\vspace{-0.2cm}
 		\caption{$\jumbo$ and how does it overcomes authenticator bottleneck.}	
 		\label{fig:dumbo-ng}
 	\end{center}
 	\vspace{-0.3cm}
 \end{figure}

\subsection{Analysis of $\jumbo$}

$\jumbo$   securely realizes the notion of  asynchronous BFT consensus (atomic broadcast) with fairness. More formally,
 
\begin{theorem}
	$\jumbo$    implements asynchronous BFT atomic broadcast with fairness except with negligible probability.
	\label{theorem:ng++}
\end{theorem}

 \smallskip
 \noindent
 {\bf Security analysis.} The   security proof of $\jumbo$   largely coincides with that of $\dumbong$ \cite{gao2022dumbo}, except that (i) the proof of liveness has to be adapted to consider possible ``stuck'' caused by hanging of votes and (ii) fairness requires new analysis. We   provide a proof sketch hereunder and defer the   full security proof to Appendix \ref{app:jumbo} in Supplementary.
	
	For proving liveness, we need to consider the pending of votes due to fairness patch.
	The roadmap of  proofs is as: 
	
	\begin{itemize}
		\item {\em All honest nodes can obtain a valid input to the current $\mvba$ protocol in constant rounds}. 
		All honest nodes won't halt their voting  in the first round of each broadcast,
		as they always continue on voting if $\delta_j=0$.	
		So every honest node can at least receives one more broadcast with valid QC from each honest broadcast sender in constant rounds.
		Given the fact that there are at least $n-f$ honest broadcast  senders, all honest nodes can prepare a valid $\mvba$ input in constant rounds. 
		\item {\em An honest node can broadcast transactions with valid QCs to all honest nodes in constant rounds}. 
		When an honest node broadcasts a batch of transactions, there are two possible cases: (i) one case is that no  honest node halts voting in the broadcast; (ii) the other case is that some nodes temporarily hang on voting in the   broadcast. In the first case, the sender can generate a valid QC for the broadcasted transactions in constant rounds. In the second case, it on average takes a constant number of $\mvba$s to solicit the highest QC of the hanging broadcast into some block,
		after which, 
		the honest nodes  would continue to vote in the broadcast,  thus allowing the broadcast sender to generate a new QC on the broadcast transactions.
		\item {\em Any honest node's broadcasted transactions would be  solicited by some $\mvba$ after expected constant rounds}. 
		Putting the above two assertions together, if a transaction is broadcasted by an honest node, 
		it first takes expected constant rounds to generate a valid QC on the transaction,
		then one round is used to send the valid QC to all honest nodes,
		and finally, QC of this transaction (or the  same sender's another QC with higher slot number) 
		would be solicited by some $\mvba$ output, with an overwhelming probability of $1-(1-q)^k$ after $k$ $\mvba$ executions. 
		Thus,   it takes   an  expected constant number of rounds to confirm   transactions  broadcasted by any honest node.
	\end{itemize}

 	For safety, the security intuitions are very straightforward and actually coincides with the original $\dumbong$ protocol:
 	
 	\begin{itemize}
 		\item {\em Guarantee   consistency of broadcasts via QCs}. A valid broadcast QC associating with slot number $s$ implies: no malicious sender can broadcast different transactions to distinct honest nodes associating with a slow number $\le s$;
 		\item {\em Ensure   agreement of each block by  soliciting   same QCs}. The agreement and external validity of $\mvba$ further ensures:  all honest nodes must solicit broadcasts with the same slot numbers and valid QCs into each   block. Thus, in combination of broadcasts' consistency, any two honest nodes'  output blocks must include   same transactions.
 	\end{itemize}
 	
	For fairness, it can be intuitively seen from ``speed limit'' placed on fastest broadcasts. More detailedly,   rationales are:
 
	\begin{itemize}
		\item {\em Bound the number of honest broadcasts solicited in each block}. Each block solicits   broadcasts from at least $n$-$f$  distinct senders, out of which     at least $f$+$1$   are honest.
		\item {\em Bound the difference in progress between malicious and honest broadcasts}. 
		Due to our fairness patch, 
		if   malicious broadcasts solicited by a block contribute  $X$ transactions,
		then  honest broadcasts solicited by the block   contribute  $\beta\cdot X$ transactions.
		Thus,  in each output block including $K$  transactions, there are at least $K \cdot \beta / (1+\beta)$  transactions   proposed by honest nodes. This gives a fairness lower bound $\alpha = \beta / (1+\beta)$. Considering $\beta$ can be a parameter close to 1, the fairness  $\alpha$   can arbitrarily approach 1/2.

		
	\end{itemize}

 \noindent
{\bf Complexity analysis}. The round, amortized communication, and message complexities of $\jumbo$ are same to these of original $\dumbong$.  For the fined-grained metric of authenticator complexity, 
  $\jumbo$ is $\bigO(n^2)$, reducing another $\bigO(n)$ order when compared to  $\dumbong$ using aggregate BLS signatures. The authenticator complexity can be divided into two parts. First, in broadcast process, only the leader generates QCs and multicasts them to receivers. Thus, for any transaction $tx$, only $\bigO(1)$ multicasts of a single QC involve. Second, in block agreement process, every node leverages provable dispersal protocol to spread out $n$ QCs, 
and later, each node reconstructs at most $\bigO(n)$ QCs, 
resulting in an authenticator complexity of $\bigO(n^2)$. Moreover, the underlying $\mvba$ protocol that we use is sMVBA from Guo  et al. \cite{guo2022speeding}, which also has an authenticator complexity of $\bigO(n^2)$.
In sum, the authenticator complexity of $\jumbo$ is $\bigO(n^2)$, which is dominated by the block agreement process.


\section{Additional Related  Work}\label{app:related}


\smallskip
\noindent
{\bf Sharding is feasible only with more scalable   consensus for each shard}. One might   wonder whether   sharding   \cite{luu2016secure,kokoris2018omniledger,zamani2018rapidchain} can overcome our scalability challenge.
But regretfully, this idea still requires each shard to execute consensus consisting of sufficiently many nodes to bound error probability. Considering a shard   of $k$  nodes randomly elected from an infinite set   with 80\% honesty,
if we expect each shard's error chance 
to be $10^{-6}$ (same parameter   in \cite{kokoris2018omniledger,zamani2018rapidchain}),
the shard shall contain $250$ or  more nodes. Namely, sharding  becomes feasible only if we   have  consensus   for hundreds of nodes as basis, which further motivates us to extend the appealing performance of   state-of-the-art asynchronous BFT     into larger scales.

\smallskip
\noindent
{\bf Dual-mode asynchronous BFT   inherits    scalability issues}.
{Recently,   a few  dual-mode asynchronous BFT   protocols \cite{bullshark,gelashvili2021jolteon,lu2021bolt} are proposed, such that they can perform similar to partially-synchronous protocols in synchrony. However, their pessimistic paths   resemble existing single-mode  asynchronous BFT protocols. For example, \cite{bullshark} almost directly inherits $\tusk$ and \cite{lu2021bolt} straightforwardly used $\dumbo$ as their pessimistic paths.
	That means,   these dual-mode protocols might still suffer from   severe scalability issues in their pessimistic case.}

\smallskip
\noindent
{\bf Executing BFT among an elected sub-committee}. 
It is also seemingly   straightforward   to sample a sufficiently small sub-committee \cite{gilad2017algorand} to solve the scalability challenge in asynchronous BFT. However,  this approach has the same problem as the sharding paradigm:
when we sample a sub-committee, we either have to assume the candidates have an exaggerated honest portion or need to sample   sufficient many nodes. For example, some reasonable parameters \cite{kokoris2018omniledger,zamani2018rapidchain} could be sampling 250 nodes out of   80\% honest candidates.

\smallskip
\noindent
{\bf Scalability solutions to partially synchronous BFT}. A few recent studies \cite{sbft,amir2008steward,avarikioti2020fnf} adopt  threshold signature towards scaling up partially synchronous BFT consensus. 
In HotStuff \cite{yin2018hotstuff-full}, the paper theoretically realizes linear communication complexity by implementing QCs from threshold signatures, but the authors'  implementation chooses ECDSA as concrete instantiation (which now can be understood by our analysis in Section \ref{sec:qc}, because implementing QCs via threshold signature in the sub-optimal individual verification manner is extremely computationally costly).
What's more, even if one uses the optimized batch verification approach in lieu of individual verification, threshold BLS signature still places  a heavy aggregation cost in the critical path (making overall computation latency not better than ECDSA), as a result of interpolation in exponent. In contrast, the aggregation of BLS multi-signature   simply ``adds'' all receives signatures. The only extra cost of BLS multi-signature compared to its threshold counterpart is $n$ bits to encode the identities of signers while embedding into QCs, which is essentially a small overhead in realistic BFT systems where $n$ is a number like at most a few hundreds.

\begin{figure*}
	\vspace{-0.1cm}
	\begin{center}
		\includegraphics[width=18.6cm]{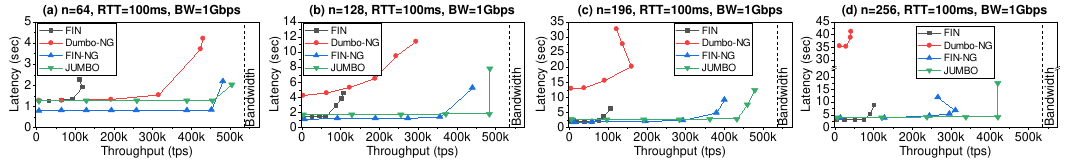}\\
		\vspace{-0.35cm}
		\caption{Latency v.s. TPS in the regional setting (100 ms RTT and 1 Gbps bandwidth)  for $n=$ 64, 128, 196 and 256   nodes.}
		\label{fig:WAN-1}
	\end{center}
	\vspace{-0.4cm}
\end{figure*}

\begin{figure*} 
	\vspace{-0.1cm}
	\begin{center}
		\includegraphics[width=18.6cm]{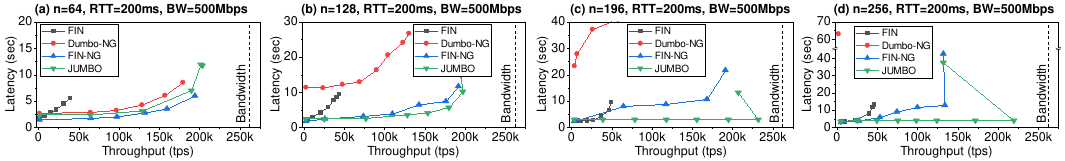}\\
		\vspace{-0.35cm}
		\caption{Latency v.s. TPS in the continental setting (200 ms RTT and 500 Mbps bandwidth) for $n=$ 64, 128, 196 and 256 nodes.}
		\label{fig:WAN-2}
	\end{center}
	\vspace{-0.4cm}
\end{figure*}

\begin{figure*} 
	\vspace{-0.1cm}
	\begin{center}
		\includegraphics[width=18.6cm]{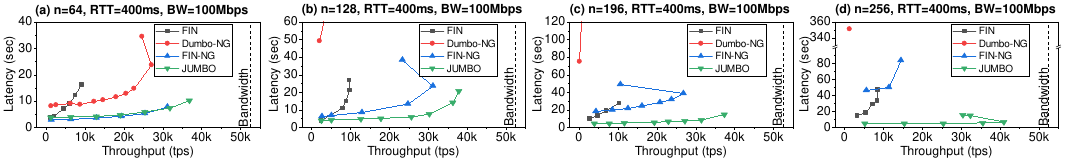}\\
		\vspace{-0.35cm}
		\caption{Latency v.s. TPS in the global setting (400 ms RTT and 100 Mbps bandwidth) for $n=$ 64, 128, 196 and 256 nodes.}
		\label{fig:WAN-3}
	\end{center}
	\vspace{-0.4cm}
\end{figure*}

\section{Implementations and Evaluations}\label{sec:eval}

We implement  $\finng$ and $\jumbo$  in Golang.
 $\dumbong$ and FIN are also implemented in the same language with same libraries and network layer. 
We then extensively evaluate  them in large  testnets   of up to 256 nodes.\footnote{The performance of open-sourced $\tusk$ implementation \cite{tusk} would be bounded by the I/O   of EC2 c6a.2xlarge instance's low-profile disk, which is much slower than bandwidth. Thus, directly comparing $\tusk$ with our implementations could be meaningless. 
	Nevertheless,   
	the original $\dumbong$   and $\tusk$ have   similar trend of performance degradation in   network scales. 
}

\smallskip
\noindent
{\bf Implementations details.} We implement  four types of processes: a client, a txpool, a broadcaster, and an orderer. 
The client process can generate and submit transactions at a fixed workload rate. The txpool process receives transactions and batches them into broadcast input.
The broadcaster processes  disseminate the batch of transactions. The orderer process decides how to cut broadcasted transactions into blocks. Each transaction is 250 byte to approximate the size of a typical Bitcoin transaction with one input and two outputs.\footnote{The size choice of transactions essentially has little impact on   evaluations, if we count  throughput in the unit of ``bytes per second'' instead. One can simply think each transaction as an alternative information unit.} We also set a upper limit of   broadcast  batch   at 4000, so a broadcast might diffuse 1-4000 transactions depending on system load.  To realize reliable fully-meshed asynchronous point-to-point channels, we implement   (persistent)   TCP connection  between every two nodes. If a TCP connection is dropped, our implementation would attempt to re-connect.

Regarding signature, $\dumbong$ \cite{gao2022dumbo} chooses   ECDSA library from Bitcoin \cite{bitcoin-core} on secp256k1 curve, and $\jumbo$ uses BLS signature library  \cite{bls} on BLS12-381 curve. 
For transaction broadcast, $\dumbong$ and $\jumbo$ choose the same QC-based broadcast,   $\finng$ variant  employs w$\rbc$ as described in Section \ref{sec:jumbo}, and FIN adopts CT05 $\rbc$ \cite{cachin2005asynchronous}. Regarding   $\mvba$, $\dumbong$  directly uses s$\mvba$ \cite{guo2022speeding}, $\jumbo$ employs s$\mvba$ extended by $\dumbo$-$\mvba$ \cite{lu2020dumbo} to fit   large input, and $\finng$ adopts our quality-enhanced version of FIN-$\mvba$---FIN-$\mvba$-Q.

\smallskip
\noindent
{\bf Evaluation setups on Amazon EC2.} We used Amazon's EC2 c6a.2xlarge instances (8 vCPUs and 16 GB main memory) in one region (Virginia) to create   testnets. 
The performances of BFT protocols were evaluated with varying scales at $n$$=$64, 128, 196, and 256 nodes, where each node corresponds to an EC2 c6a.2xlarge instance.
To  simulate  WAN environment reproducibly and affordably, we use Traffic Control (TC) tool of Linux to configure (upload) bandwidth and latency of all nodes.\footnote{Though simulating WAN in one AWS region using TC can greatly save the expense as no data transfer fee,   our    experiments'  cost is still  about \$6,000.}
We consider 3 deployment scenarios: global, continental, and regional. The global setting reflects a globally distributed environment with 400 ms round-trip time (RTT) and 100 Mbps bandwidth. The continental setting is with   200 ms RTT and 500 Mbps bandwidth to depict the network condition inside a large country or a union  of countries. The regional setting has  100 ms RTT and 1 Gbps bandwidth.

\smallskip
\noindent
{\bf Basic benchmarking.}  
Figures \ref{fig:WAN-1}, \ref{fig:WAN-2} and \ref{fig:WAN-3} 
shows throughput and latency of each   protocol  on varying input rates and scales,
for the regional, continental and global settings, respectively.
Each data point is averaged over  at least  a minute execution after a warm-up period of one block. 

\smallskip 
\noindent
\underline{\smash{\emph{Peak throughput}}}. $\jumbo$ always presents the largest peak throughput.
More intriguingly, its peak throughput  is close to the line speed, despite $n$ or network settings.
For example, when bandwidth is 500 Mbps, its throughput can roughly achieve 250,000 tx/s, which approximates 500 Mbps as we use 2000-bit (256-byte) transaction.\footnote{Note 500 Mbps  = $500*1024^2$ bit/s instead of   $5*10^8$ bit/s in Linux TC. 
}
In addition, $\finng$'s peak throughput  is larger than that of FIN and  $\dumbong$  in nearly all cases, although it is less than   $\jumbo$.

The peak throughput of $\jumbo$ can exceed  that of $\finng$ because the QC-based broadcasts in $\jumbo$ can better utilize bandwidth resources due to less   overheads: First, the ``uncontributive'' network packet headers in $\jumbo$ are an order-of-$\bigO(n)$ less, because of its asymptotically better message complexity; Second,   ``loadless'' message passing rounds are also reduced in   $\jumbo$, considering that each  broadcast in $\jumbo$ only has a single loadless round but each w$\rbc$ in $\finng$ has 2-3 loadless rounds.


\noindent
\underline{\smash{\emph{Latency-throughput trade-offs}}}.  It is clear that
$\jumbo$ strictly surpasses $\finng$  in the global and continental settings when $n$$\ge$$128$ nodes, as $\jumbo$   can attain the same throughput at a lower latency.
What's even more intriguing is that $\jumbo$'s latency can remain  low while its  throughput increases from minimum to maximum. 
This hints at that $\jumbo$  is promising to simultaneously handle throughput-critical and latency-critical applications. 
$\finng$ cannot match $\jumbo$ in handling substantial system loads in large network such as $n$=$256$, but it might exhibit lower latency 
when bandwidth is sufficient (e.g., 1 Gbps) and/or system is relatively small (e.g., 64 nodes), 
indicating its applications in   scenarios that are extremely latency-sensitive with low system loads and adequate network resources.
Again, FIN and  $\dumbong$ perform much worse than $\finng$ and $\jumbo$ in most cases.

\smallskip
\noindent
{\bf Benchmarking with faulty nodes.} We also   consider 3 types of faulty nodes: (i) crash nodes, (ii) Byzantine nodes that send invalid  signatures to cause honest nodes verify all signatures before blocklist comes into effect, and (iii)  Byzantine nodes that can lower the quality of $\mvba$ to downgrade performance. 
These malicious setting experiments are conducted in a network of 1 Gbps bandwidth and 100 ms RTT.

\begin{figure}[H]
	\vspace{-0.3cm}
	\begin{center}
		\includegraphics[width=7.1cm]{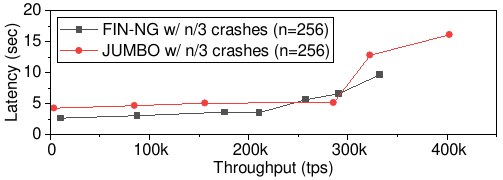}\\
		\vspace{-0.3cm}
		\caption{Latency v.s. throughput in case of 100ms RTT and 1Gbps bandwidth for $n=$256 nodes ({with 85   crashed nodes}).}
		\label{fig:crash}
	\end{center}
	\vspace{-0.35cm}
\end{figure}

\begin{figure}[H]
	\vspace{-0.2cm}
	\begin{center}
		\includegraphics[width=7.5cm]{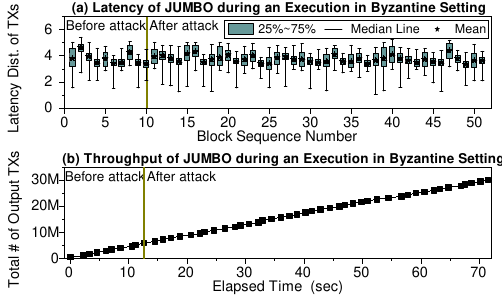}\\
		\vspace{-0.3cm}
		\caption{Execution of $\jumbo$    for $n=$256 nodes (where are 85 malicious nodes   sending invalid signatures).}
		\label{fig:byz}
	\end{center}
	\vspace{-0.3cm}
\end{figure}

\begin{figure}[H]
	\vspace{-0.2cm}
	\begin{center}
		\includegraphics[width=7.1cm]{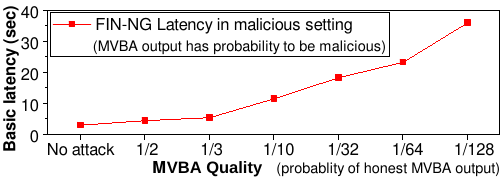}\\
		\vspace{-0.2cm}
		\caption{Latency of FIN-NG v.s. the quality of MVBA used in FIN-NG for $n=$256 nodes ({\bf with 85   malicious nodes}).}
		\label{fig:latency-fairness}
	\end{center}
	\vspace{-0.3cm}
\end{figure}

\noindent
\underline{\smash{\emph{Crash faults}}}. As  illustrated in Fig  \ref{fig:crash}, we evaluated the performance of  $\jumbo$  and $\finng$ with $\lfloor (n-1)/3 \rfloor$ crash nodes for $n=$256   nodes. Crashes cause the latency of $\finng$ and $\jumbo$ slightly increased. 
The peak throughputs of $\jumbo$  and $\finng$ are almost not impacted by crashes,
and they still can reasonably track    available network bandwidth. 

\noindent
\underline{\smash{\emph{Byzantine faults that send invalid digital signatures}}}. We evaluated $\jumbo$ for $n$=$256$   with 85 malicious nodes that send  fake signatures to  downgrade the performance of   batch-verification. Figure \ref{fig:byz} (a) and (b) plot the latency and throughput, respectively,  during an execution where the Byzantine nodes launch the attack while running the 11th and later blocks. Clearly, the attack almost causes negligible impact on these performance metrics, as the Byzantine nodes were quickly blocklisted and therefore their attack was prevented.

\noindent
\underline{\smash{\emph{Byzantine faults that manipulate    $\mvba$ output quality}}}. Then we run tests of $\finng$ for $n$=$256$   with 85 Byzantine nodes that attempt to manipulate $\mvba$ output.
These Byzantine nodes' inputs to $\mvba$ have minimum validity (e.g., only solicit $n-f$ broadcasts from distinct nodes), 
and they attempt to let $\mvba$ decide a malicious output as much as possible.
As Figure \ref{fig:latency-fairness} plots, we compared the effect of above attack while $\mvba$'s quality varies from 1/2 to 1/128.
It is obvious to see that 1/2-quality is more robust against such attack and renders a latency closest to the  case of no such attack.

\noindent
{\bf Benchmarking in  fluctuating network}.
We also test $\finng$ and $\jumbo$ with 256 nodes in a dynamic network controlled by Linux TC    to   switch  between 1 Gbps and 500 Mbps for each 15 sec. The evaluated results are plotted in Figure \ref{fig:changing-network}.
 $\jumbo$   can always closely track  the network capacity and achieve throughput close to the instantaneous line rate. $\finng$ is also  can switch  smoothly across varying network conditions, though its concrete throughput is less than $\jumbo$. 

\begin{figure}[H]
	\vspace{-0.2cm}
	\begin{center}
		\includegraphics[width=7cm]{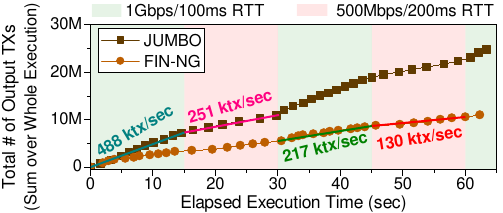}\\
		\vspace{-0.2cm}
		\caption{Executions of $\jumbo$ and $\finng$ in a fluctuating network that changes   bandwidth and RTT for every 15 sec.}
		\label{fig:changing-network}
	\end{center}
	\vspace{-0.3cm}
\end{figure}

\section{Conclusion}

We present a couple of more scalable asynchronous BFT consensus protocols  $\finng$ and $\jumbo$ with
reduced authenticator overhead. 
Extensive experiments are conducted, revealing that these scalable  designs can  performantly operate in large-scale WAN settings with up to 256 nodes. 

\bibliographystyle{IEEEtran}
{\footnotesize
\bibliography{references}

\begin{thebibliography}{10}
\providecommand{\url}[1]{#1}
\csname url@samestyle\endcsname
\providecommand{\newblock}{\relax}
\providecommand{\bibinfo}[2]{#2}
\providecommand{\BIBentrySTDinterwordspacing}{\spaceskip=0pt\relax}
\providecommand{\BIBentryALTinterwordstretchfactor}{4}
\providecommand{\BIBentryALTinterwordspacing}{\spaceskip=\fontdimen2\font plus
\BIBentryALTinterwordstretchfactor\fontdimen3\font minus
  \fontdimen4\font\relax}
\providecommand{\BIBforeignlanguage}[2]{{%
\expandafter\ifx\csname l@#1\endcsname\relax
\typeout{** WARNING: IEEEtran.bst: No hyphenation pattern has been}%
\typeout{** loaded for the language `#1'. Using the pattern for}%
\typeout{** the default language instead.}%
\else
\language=\csname l@#1\endcsname
\fi
#2}}
\providecommand{\BIBdecl}{\relax}
\BIBdecl

\bibitem{ben1983another}
M.~Ben-Or, ``Another advantage of free choice (extended abstract) completely
  asynchronous agreement protocols,'' in \emph{Proc. PODC 1983}, pp. 27--30.

\bibitem{rabin1983randomized}
M.~O. Rabin, ``Randomized byzantine generals,'' in \emph{Proc. FOCS 1983}, pp.
  403--409.

\bibitem{MMR15}
A.~Mostefaoui, H.~Moumen, and M.~Raynal, ``Signature-free asynchronous
  byzantine consensus with t< n/3 and o (n2) messages,'' in \emph{Proc. PODC
  2014}, 2014, pp. 2--9.

\bibitem{canetti1993fast}
R.~Canetti and T.~Rabin, ``Fast asynchronous byzantine agreement with optimal
  resilience,'' in \emph{Proc. STOC 1993}, pp. 42--51.

\bibitem{cachin2002secure}
C.~Cachin and J.~A. Poritz, ``Secure intrusion-tolerant replication on the
  internet,'' in \emph{Proc. DSN 2002}, 2002, pp. 167--176.

\bibitem{moniz2008ritas}
H.~Moniz, N.~F. Neves, M.~Correia, and P.~Verissimo, ``Ritas: Services for
  randomized intrusion tolerance,'' \emph{IEEE transactions on dependable and
  secure computing}, vol.~8, no.~1, pp. 122--136, 2008.

\bibitem{patra2009simple}
A.~Patra, A.~Choudhary, and C.~Pandu~Rangan, ``Simple and efficient
  asynchronous byzantine agreement with optimal resilience,'' in \emph{Proc.
  PODC 2009}, pp. 92--101.

\bibitem{abraham2008almost}
I.~Abraham, D.~Dolev, and J.~Y. Halpern, ``An almost-surely terminating
  polynomial protocol for asynchronous byzantine agreement with optimal
  resilience,'' in \emph{Proc. PODC 2008}, pp. 405--414.

\bibitem{cachin00}
C.~Cachin, K.~Kursawe, and V.~Shoup, ``Random oracles in constantinople:
  Practical asynchronous byzantine agreement using cryptography,''
  \emph{Journal of Cryptology}, vol.~18, no.~3, pp. 219--246, 2005.

\bibitem{ben2003resilient}
M.~Ben-Or and R.~El-Yaniv, ``Resilient-optimal interactive consistency in
  constant time,'' \emph{Distributed Computing}, vol.~16, no.~4, pp. 249--262,
  2003.

\bibitem{cachin2001secure}
C.~Cachin, K.~Kursawe, F.~Petzold, and V.~Shoup, ``Secure and efficient
  asynchronous broadcast protocols,'' in \emph{CRYPTO 2001}, pp. 524--541.

\bibitem{miller2016honey}
A.~Miller, Y.~Xia, K.~Croman, E.~Shi, and D.~Song, ``The honey badger of bft
  protocols,'' in \emph{Proc. CCS 2016}, 2016, pp. 31--42.

\bibitem{beat}
S.~Duan, M.~K. Reiter, and H.~Zhang, ``Beat: Asynchronous bft made practical,''
  in \emph{Proc. CCS 2018}, 2018, pp. 2028--2041.

\bibitem{guo2020dumbo}
B.~Guo, Z.~Lu, Q.~Tang, J.~Xu, and Z.~Zhang, ``Dumbo: Faster asynchronous bft
  protocols,'' in \emph{Proc. CCS 2020}, pp. 803--818.

\bibitem{abraham2018validated}
I.~Abraham, D.~Malkhi, and A.~Spiegelman, ``Asymptotically optimal validated
  asynchronous byzantine agreement,'' in \emph{Proc. PODC 2019}, pp. 337--346.

\bibitem{tusk}
G.~Danezis, L.~Kokoris-Kogias, A.~Sonnino, and A.~Spiegelman, ``Narwhal and
  tusk: a dag-based mempool and efficient bft consensus,'' in \emph{Proc.
  EuroSys 2022}, pp. 34--50.

\bibitem{lu2020dumbo}
Y.~Lu, Z.~Lu, Q.~Tang, and G.~Wang, ``Dumbo-mvba: Optimal multi-valued
  validated asynchronous byzantine agreement, revisited,'' in \emph{Proc. PODC
  2020}, pp. 129--138.

\bibitem{gao2022dumbo}
Y.~Gao, Y.~Lu, Z.~Lu, Q.~Tang, J.~Xu, and Z.~Zhang, ``Dumbo-ng: Fast
  asynchronous bft consensus with throughput-oblivious latency,'' in
  \emph{Proc. CCS 2022}, pp. 1187--1201.

\bibitem{abraham2021reaching}
I.~Abraham, P.~Jovanovic, M.~Maller, S.~Meiklejohn, G.~Stern, and A.~Tomescu,
  ``Reaching consensus for asynchronous distributed key generation,'' in
  \emph{Proc. PODC 2021}, pp. 363--373.

\bibitem{duan2023practical}
S.~Duan, X.~Wang, and H.~Zhang, ``Practical signature-free asynchronous common
  subset in constant time,'' in \emph{Proc. CCS 2023}.

\bibitem{yang2021dispersedledger}
L.~Yang, S.~J. Park, M.~Alizadeh, S.~Kannan, and D.~Tse, ``{DispersedLedger}:
  {High-Throughput} byzantine consensus on variable bandwidth networks,'' in
  \emph{Proc. NSDI 2022}.

\bibitem{FLP85}
M.~J. Fischer, N.~A. Lynch, and M.~S. Paterson, ``Impossibility of distributed
  consensus with one faulty process,'' \emph{Journal of the ACM (JACM)},
  vol.~32, no.~2, pp. 374--382, 1985.

\bibitem{bitcoin}
S.~Nakamoto, ``Bitcoin: A peer-to-peer electronic cash system,''
  \emph{Decentralized Business Review}, p. 21260, 2008.

\bibitem{pbft}
M.~Castro, B.~Liskov \emph{et~al.}, ``Practical byzantine fault tolerance,'' in
  \emph{Proc. OSDI 1999}, pp. 173--186.

\bibitem{saad2021revisiting}
M.~Saad, A.~Anwar, S.~Ravi, and D.~Mohaisen, ``Revisiting nakamoto consensus in
  asynchronous networks: A comprehensive analysis of bitcoin safety and
  chainquality,'' in \emph{Proc. CCS 2021}, pp. 988--1005.

\bibitem{yin2018hotstuff-full}
M.~Yin, D.~Malkhi, M.~K. Reiter, G.~G. Gueta, and I.~Abraham, ``Hotstuff: Bft
  consensus with linearity and responsiveness,'' in \emph{Proc. PODC 2019},
  2019, pp. 347--356.

\bibitem{guo2022speeding}
B.~Guo, Y.~Lu, Z.~Lu, Q.~Tang, J.~Xu, and Z.~Zhang, ``Speeding dumbo: Pushing
  asynchronous bft closer to practice,'' in \emph{Proc. NDSS 2022}.

\bibitem{bft-smart}
A.~Bessani, J.~Sousa, and E.~E. Alchieri, ``State machine replication for the
  masses with bft-smart,'' in \emph{Proc. DSN 2014}, pp. 355--362.

\bibitem{li2023performance}
Z.~Li, A.~Sonnino, and P.~Jovanovic, ``Performance of eddsa and bls signatures
  in committee-based consensus,'' in \emph{Proc. ApPLIED 2023}.

\bibitem{avarikioti2020fnf}
Z.~Avarikioti, L.~Heimbach, R.~Schmid, L.~Vanbever, R.~Wattenhofer, and
  P.~Wintermeyer, ``Fnf-bft: Exploring performance limits of bft protocols,''
  \emph{arXiv preprint arXiv:2009.02235}, 2020.

\bibitem{zhang2022pace}
H.~Zhang and S.~Duan, ``Pace: Fully parallelizable bft from reproposable
  byzantine agreement,'' in \emph{Proc. CCS 2022}, pp. 3151--3164.

\bibitem{benor}
M.~Ben-Or, B.~Kelmer, and T.~Rabin, ``Asynchronous secure computations with
  optimal resilience,'' in \emph{Proc. PODC 1994}, pp. 183--192.

\bibitem{cachin05}
H.~V. Ramasamy and C.~Cachin, ``Parsimonious asynchronous
  byzantine-fault-tolerant atomic broadcast,'' in \emph{Proc. OPODIS 2005}, pp.
  88--102.

\bibitem{merkle1987digital}
R.~C. Merkle, ``A digital signature based on a conventional encryption
  function,'' in \emph{Conference on the theory and application of
  cryptographic techniques}, 1987, pp. 369--378.

\bibitem{reed1960polynomial}
I.~S. Reed and G.~Solomon, ``Polynomial codes over certain finite fields,''
  \emph{Journal of the society for industrial and applied mathematics}, vol.~8,
  no.~2, pp. 300--304, 1960.

\bibitem{rs-lib}
\BIBentryALTinterwordspacing
``Reed-solomon erasure coding in go,'' accessed: 2024-3-10. [Online].
  Available: \url{{https://github.com/klauspost/reedsolomon}}
\BIBentrySTDinterwordspacing

\bibitem{blomer1995xor}
J.~Blomer, M.~Kalfane, M.~Karpinski, R.~Karp, M.~Luby, and D.~Zuckermank, ``An
  xor-based erasure-resilient coding scheme,'' \emph{Tech Report, Tech. Rep.},
  1995.

\bibitem{plank2005optimizing}
J.~S. Plank, ``Optimizing cauchy reed-solomon codes for fault-tolerant storage
  applications,'' \emph{University of Tennessee, Tech. Rep. CS-05-569},
  vol.~38, 2005.

\bibitem{mirbft}
C.~Stathakopoulou, T.~David, M.~Pavlovic, and M.~Vukoli{\'c}, ``Mir-bft:
  High-throughput robust bft for decentralized networks,'' \emph{arXiv preprint
  arXiv:1906.05552}, 2019.

\bibitem{schnorr1991efficient}
C.-P. Schnorr, ``Efficient signature generation by smart cards,'' \emph{Journal
  of cryptology}, vol.~4, pp. 161--174, 1991.

\bibitem{chalkias2021non}
K.~Chalkias, F.~Garillot, Y.~Kondi, and V.~Nikolaenko, ``Non-interactive
  half-aggregation of eddsa and variants of schnorr signatures,'' in
  \emph{Cryptographers' Track at the RSA Conference}, 2021, pp. 577--608.

\bibitem{chen2022half}
Y.~Chen and Y.~Zhao, ``Half-aggregation of schnorr signatures with tight
  reductions,'' in \emph{ESORICS 2022}, pp. 385--404.

\bibitem{boneh2004short}
D.~Boneh, B.~Lynn, and H.~Shacham, ``Short signatures from the weil pairing,''
  \emph{Journal of cryptology}, vol.~17, pp. 297--319, 2004.

\bibitem{boneh2018compact}
D.~Boneh, M.~Drijvers, and G.~Neven, ``Compact multi-signatures for smaller
  blockchains,'' in \emph{ASIACRYPT 2018}, pp. 435--464.

\bibitem{boldyreva2003threshold}
A.~Boldyreva, ``Threshold signatures, multisignatures and blind signatures
  based on the gap-diffie-hellman-group signature scheme,'' in
  \emph{International Workshop on Public Key Cryptography}, 2003, pp. 31--46.

\bibitem{opticoin}
\BIBentryALTinterwordspacing
``Use threshold signature as source of randomness,'' accessed: 2023-7-20.
  [Online]. Available:
  \url{{https://github.com/amiller/HoneyBadgerBFT/commit/ac787c486bc8def0fa0da3cb3287a12926bd437a}}
\BIBentrySTDinterwordspacing

\bibitem{jannes2023beaufort}
K.~Jannes, E.~H. Beni, B.~Lagaisse, and W.~Joosen, ``Beaufort: Robust byzantine
  fault tolerance for client-centric mobile web applications,'' \emph{IEEE
  TPDS}, vol.~34, no.~4, pp. 1241--1252, 2023.

\bibitem{renavss}
S.~Das, Z.~Xiang, and L.~Ren, ``Asynchronous data dissemination and its
  applications,'' in \emph{Proc. CCS 2021}, pp. 2705--2721.

\bibitem{alhaddad2022balanced}
N.~Alhaddad, S.~Das, S.~Duan, L.~Ren, M.~Varia, Z.~Xiang, and H.~Zhang,
  ``Balanced byzantine reliable broadcast with near-optimal communication and
  improved computation,'' in \emph{Proc. PODC 2022}, pp. 399--417.

\bibitem{bracha1987asynchronous}
G.~Bracha, ``Asynchronous byzantine agreement protocols,'' \emph{Information
  and Computation}, vol.~75, no.~2, pp. 130--143, 1987.

\bibitem{bracha1984asynchronous}
G.~Bracha, ``An asynchronous [(n-1)/3]-resilient consensus protocol,'' in
  \emph{Proc. PODC 1984}, pp. 154--162.

\bibitem{goren2022probabilistic}
G.~Goren, Y.~Moses, and A.~Spiegelman, ``Probabilistic indistinguishability and
  the quality of validity in byzantine agreement,'' in \emph{Proc. AFT}, 2022,
  pp. 111--125.

\bibitem{libert2011adaptively}
B.~Libert and M.~Yung, ``Adaptively secure non-interactive threshold
  cryptosystems,'' in \emph{Proc. ICALP 2011}, 2011, pp. 588--600.

\bibitem{kokoris2020asynchronous}
E.~Kokoris~Kogias, D.~Malkhi, and A.~Spiegelman, ``Asynchronous distributed key
  generation for computationally-secure randomness, consensus, and threshold
  signatures.'' in \emph{Proc. CCS 2020}, pp. 1751--1767.

\bibitem{das2022practical}
S.~Das, T.~Yurek, Z.~Xiang, A.~Miller, L.~Kokoris-Kogias, and L.~Ren,
  ``Practical asynchronous distributed key generation,'' in \emph{IEEE S\&P
  2022}, pp. 2518--2534.

\bibitem{gao2021efficient}
Y.~Gao, Y.~Lu, Z.~Lu, Q.~Tang, J.~Xu, and Z.~Zhang, ``Efficient asynchronous
  byzantine agreement without private setups,'' in \emph{Proc. ICDCS 2022}.

\bibitem{garay2015bitcoin}
J.~Garay, A.~Kiayias, and N.~Leonardos, ``The bitcoin backbone protocol:
  Analysis and applications,'' in \emph{EUROCRYPT 2015}, pp. 281--310.

\bibitem{dag}
I.~Keidar, E.~Kokoris-Kogias, O.~Naor, and A.~Spiegelman, ``All you need is
  dag,'' in \emph{Proc. PODC 2021}, pp. 165--175.

\bibitem{cachin2005asynchronous}
C.~Cachin and S.~Tessaro, ``Asynchronous verifiable information dispersal,'' in
  \emph{Proc. SRDS 2005}, pp. 191--201.

\bibitem{yurek2023long}
T.~Yurek, Z.~Xiang, Y.~Xia, and A.~Miller, ``Long live the honey badger: Robust
  asynchronous $\{$DPSS$\}$ and its applications,'' in \emph{USENIX Security
  23}, pp. 5413--5430.

\bibitem{das2023practical}
S.~Das, Z.~Xiang, L.~Kokoris-Kogias, and L.~Ren, ``Practical asynchronous
  high-threshold distributed key generation and distributed polynomial
  sampling,'' in \emph{USENIX Security 23}, pp. 5359--5376.

\bibitem{abraham2023perfectly}
I.~Abraham, G.~Asharov, A.~Patra, and G.~Stern, ``Perfectly secure asynchronous
  agreement on a core set in constant expected time,'' \emph{Cryptology ePrint
  Archive}, 2023.

\bibitem{rfc2986}
\BIBentryALTinterwordspacing
M.~Nystrom and B.~Kaliski, ``{PKCS \#10: Certification Request Syntax
  Specification Version 1.7},'' RFC 2986, Nov. 2000. [Online]. Available:
  \url{https://www.rfc-editor.org/info/rfc2986}
\BIBentrySTDinterwordspacing

\bibitem{reiter1994secure}
M.~K. Reiter, ``Secure agreement protocols: Reliable and atomic group multicast
  in rampart,'' in \emph{Proc. ACM CCS 1994}, pp. 68--80.

\bibitem{optiagg}
\BIBentryALTinterwordspacing
``Pragmatic signature aggregation with bls,'' accessed: 2023-7-20. [Online].
  Available:
  \url{{https://ethresear.ch/t/pragmatic-signature-aggregation-with-bls/2105}}
\BIBentrySTDinterwordspacing

\bibitem{tomescu2020towards}
A.~Tomescu, R.~Chen, Y.~Zheng, I.~Abraham, B.~Pinkas, G.~G. Gueta, and
  S.~Devadas, ``Towards scalable threshold cryptosystems,'' in \emph{Proc. IEEE
  S\&P 2020}, pp. 877--893.

\bibitem{bitcoin-core}
\BIBentryALTinterwordspacing
libsecp256k1. [Online]. Available:
  \url{https://github.com/bitcoin-core/secp256k1}
\BIBentrySTDinterwordspacing

\bibitem{barreto2003constructing}
P.~S. Barreto, B.~Lynn, and M.~Scott, ``Constructing elliptic curves with
  prescribed embedding degrees,'' in \emph{SCN 2003}, pp. 257--267.

\bibitem{bls}
\BIBentryALTinterwordspacing
bls-go-binary. [Online]. Available:
  \url{https://github.com/herumi/bls-go-binary}
\BIBentrySTDinterwordspacing

\bibitem{kim2016extended}
T.~Kim and R.~Barbulescu, ``Extended tower number field sieve: A new complexity
  for the medium prime case,'' in \emph{Annual international cryptology
  conference}, 2016, pp. 543--571.

\bibitem{luu2016secure}
L.~Luu, V.~Narayanan, C.~Zheng, K.~Baweja, S.~Gilbert, and P.~Saxena, ``A
  secure sharding protocol for open blockchains,'' in \emph{Proc. ACM CCS
  2016}, pp. 17--30.

\bibitem{kokoris2018omniledger}
E.~Kokoris-Kogias, P.~Jovanovic, L.~Gasser, N.~Gailly, E.~Syta, and B.~Ford,
  ``Omniledger: A secure, scale-out, decentralized ledger via sharding,'' in
  \emph{Proc. IEEE S\&P 2018}, pp. 583--598.

\bibitem{zamani2018rapidchain}
M.~Zamani, M.~Movahedi, and M.~Raykova, ``Rapidchain: Scaling blockchain via
  full sharding,'' in \emph{Proc. ACM CCS 2018}, pp. 931--948.

\bibitem{bullshark}
N.~Giridharan, L.~Kokoris-Kogias, A.~Sonnino, and A.~Spiegelman, ``Bullshark:
  Dag bft protocols made practical,'' in \emph{Proc. CCS 2022}, 2022.

\bibitem{gelashvili2021jolteon}
R.~Gelashvili, L.~Kokoris-Kogias, A.~Sonnino, A.~Spiegelman, and Z.~Xiang,
  ``Jolteon and ditto: Network-adaptive efficient consensus with asynchronous
  fallback,'' in \emph{FC 2022}.

\bibitem{lu2021bolt}
Y.~Lu, Z.~Lu, and Q.~Tang, ``Bolt-dumbo transformer: Asynchronous consensus as
  fast as pipelined bft,'' in \emph{Proc. CCS 2022}, 2022.

\bibitem{gilad2017algorand}
Y.~Gilad, R.~Hemo, S.~Micali, G.~Vlachos, and N.~Zeldovich, ``Algorand: Scaling
  byzantine agreements for cryptocurrencies,'' in \emph{Proceedings of the 26th
  symposium on operating systems principles}, 2017, pp. 51--68.

\bibitem{sbft}
G.~G. Gueta, I.~Abraham, S.~Grossman, D.~Malkhi, B.~Pinkas, M.~Reiter, D.-A.
  Seredinschi, O.~Tamir, and A.~Tomescu, ``Sbft: a scalable and decentralized
  trust infrastructure,'' in \emph{Proc. DSN 2019}, 2019, pp. 568--580.

\bibitem{amir2008steward}
Y.~Amir, C.~Danilov, D.~Dolev, J.~Kirsch, J.~Lane, C.~Nita-Rotaru, J.~Olsen,
  and D.~Zage, ``Steward: Scaling byzantine fault-tolerant replication to wide
  area networks,'' \emph{IEEE Transactions on Dependable and Secure Computing},
  vol.~7, no.~1, pp. 80--93, 2008.

\bibitem{neiheiser2021kauri}
R.~Neiheiser, M.~Matos, and L.~Rodrigues, ``Kauri: Scalable bft consensus with
  pipelined tree-based dissemination and aggregation,'' in \emph{Proc. SOSP
  2021}, pp. 35--48.

\end{thebibliography}
}


\appendices

\section{CPU Time of Sign, Non-interactive Aggregation and Verification of Typical Digital Signatures}\label{app:sig}

Table \ref{tab:sig} lists the signing time, verification time, aggregation time, and batch-verification time of typical digital signatures.
From Table \ref{tab:sig}, we have the next key findings:
a) QCs constructed from BLS multi-signature   can be verified much faster than QCs constructed from concatenated ECDSA signatures and half-aggregated Schnorr signatures, e.g., verifying a QC from BLS multi-signature is an order of magnitude faster than those from ECDSA and Schnorr for $n=256$; b) however, verifying a single BLS signature is much slower than verifying a single ECDSA/Schnorr signature because BLS verification needs costly pairing operations of elliptic curve, e.g., the verification of a single BLS signature is an order of magnitude slower than ECDSA and Schnorr signatures. 
Given our benchmarking of digital signatures, it becomes easy to understand why the earlier study \cite{li2023performance} concludes that the tempting BLS multi-signature has inferior performance in their BFT systems with relatively larger sizes, because the verification of $n-f$  single BLS signatures can quickly dominate the overall computing cost while $n$ increases. 

\begin{table*}
	\caption{Time of  signing, verifying a single signature, aggregating/concatenating $(n-f)$ signatures, and verifying a $(n-f)$-sized QC, where each data point is averaged over 10 thousand executions at an AWS EC2 c6a.2xlarge server. ECDSA/Schnorr are tested through library \cite{bitcoin-core}  over  Secp256k1 curve  and BLS is tested through library \cite{bls} over  BLS12-381 curve.}\label{tab:sig}
	\vspace{-0.1cm}
	\begin{footnotesize}
	\centering
		\begin{tabular}{|l|c|c|ccccccc|ccccccc|}
			\hline\rule{0pt}{11pt}
			\multirow{2}{*}{} & \multirow{2}{*}{\begin{tabular}[c]{@{}c@{}}Sign\\ ($\mu s$)\end{tabular}} & \multirow{2}{*}{\begin{tabular}[c]{@{}c@{}}Verify\\ ($\mu s$)\end{tabular}} & \multicolumn{7}{c|}{(Half-)aggregate $(n-f)$ signatures  to form QC  ($\mu s$)} & \multicolumn{7}{c|}{Verify $(n-f)$-sized QC   ($\mu s$)} \\ \cline{4-17} \rule{0pt}{11pt}
			&                                                                           &                                                                             & $n$=4      & 16     & 32      & 64     & 128    & 192    & 256    & $n$=4       & 16      & 32      & 64       & 128      & 192     & 256     \\ \hline\rule{0pt}{11pt}ECDSA             & 31                                                                     & 43                                                                       & \multicolumn{7}{c|}{0 as just concatenate}                    & \multicolumn{7}{c|}{$n-f$ times of verifying an individual signature}          \\
			Schnorr           & 25                                                                     & 44                                                                      & 18  & 63  & 120  & 245 & 485 & 731 & 970 & 141  & 430  & 787  & 1587  & 2714  & 3975 & 5005 \\
			BLS multi-sig              & 187                                                                     & 891                                                                      & 1   & 6   & 11    & 18  & 49  & 74  & 102  & 904  & 916  & 932  & 951   & 1036   & 1102  & 1184  \\

			BLS thld-sig              & 187                                                                     & 891                                                                      & 203   & 769   & 1448    & 3137  & 6294  & 9578  & 12976  & 891  & 891  & 891  & 891   & 891   & 891  & 891  \\ \hline
			
			\hline
		\end{tabular}
	\end{footnotesize}
\end{table*}

\section{Quality Downgrade Attack of FIN-MVBA}\label{app:quality-attack}
\noindent
As aforementioned in Section \ref{sec:intro} and Section \ref{sec:jumbo},   the quality of FIN-MVBA \cite{duan2023practical} could be only 1/3 under the influence of an adaptive adversary that can perform ``after-fact-removal'', which is lower than the optimal upper bound of 1/2 \cite{goren2022probabilistic}. 

Here we illustrate a concrete adversary to perform such quality degradation attack. Note that   the  adversary is adaptive such that  it can gradually corrupt up to $f$ nodes during the course of protocol. Besides, we assume the number of participating nodes $n$ is large enough in order to simplify the calculation of probability.

Below are the detailed attacking steps:
\begin{enumerate}
	\item Before the start of the FIN-MVBA protocol, the adversary chooses $f-1$  nodes to corrupt, and let us denote this corrupted set by $A$, such that the input of any node in $A$ can be replaced and thus controlled by the adversary. 
	\item At the moment when  the FIN-MVBA protocol starts, the adversary selects $f$ honest nodes, denoted by set $D$, and delays all messages related to these nodes in $D$. 
	Note that in addition to $D$, there are still $f+2$ so-far-honest nodes, and we denote them as set $H$.
	\item The adversary controls the malicious nodes in $A$ to execute the protocol like the honest ones until the $election$ protocol. 
	Such that $H$ can proceed to invoke the random leader $election$ protocol. To minimize the adversary's  ability, we assume the reconstruction threshold of $election$ is $n-f$, which equals to the size of $A \cup H$. 
	\item The adversary controls $A$ not send messages in the $election$ protocol, but it can receive $H$'s messages in the $election$ protocol, which allows $A$ to get the $election$ protocol's output $k$.
	
	\item Once the adversary gets $k$, there are three possible cases:
	\begin{enumerate}
		\item  $k$ falls in $A$. In this case, FIN-MVBA will output the input value of $P_k$, which is chosen by the adversary. Note that the probability of this scenario occurring is about 1/3.
		\item  $k$ falls in $H$. In this case, FIN-MVBA will output the input value of $P_k$, which is an input chosen by the honest node. Note that the probability of this scenario occurring is also about 1/3.
		\item  $k$ falls in $D$. In this case, the adversary can immediately corrupt node $P_k$ because it still has a quota of corruption. Then the adversary removes all messages that $P_k$ attempted to send because these messages are not delivered yet, as such the adversary can replace $P_k$'s input value, and rushingly finishes $P_k$'s $\rbc$, after which, the adversary delivers all messages, then its manipulating value will output in FIN-MVBA. So in this scenario, the output is always chosen by the adversary. Note that the probability of this scenario occurring is also about 1/3.
		\end{enumerate}
\end{enumerate}

The above steps complete the attack, and result in that with probability 2/3, the output is manipulated by the adversary.

\section{Proof of FIN-$\mvba$-Q}\label{app:mvba-proof}
\noindent
{\bf Security of avw$\rbc$}. 
Note that a few properties of avw$\rbc$, such as weak agreement and validity, can immediately inherited from $w\rbc$, 
and we will only focus on the proofs of   totality and external validity properties.

\begin{lemma}\label{lemma:totality}
	\textbf{Totality of} avw$\rbc$. For any avw$\rbc$ instance, if an honest node $wr\text{-}delivers$ some value, any honest node eventually wr-delivers some value.
\end{lemma}
\begin{proof}
	Let node $\node_i$ be the first honest node who wr-delivers \emph{h} in ${avw\rbc}_j$. So $\node_i$ must have received $n-f$ \emph{READY(h)} messages from distinct nodes, among which $f+1$ are honest nodes. As shown in Figure \ref{algo:rbc}, invoking $abandon()$ only blocks sending \emph{ECHO} messages. So these $f+1$ honest nodes will eventually transmit their \emph{READY(h)} messages to all honest nodes. Thus, any honest node will eventually receives $f+1$ \emph{READY(h)} messages, and will tansmit \emph{READY(h)} if has not sent it. Therefore, every honest node will eventually receive $n-f$ \emph{READY(h)}, and $wr\text{-}deliver$ \emph{h}.
\end{proof}

\begin{lemma}\label{lemma:ext-val}
	\textbf{External validity of} avw$\rbc$. For any avw$\rbc$ instance, If an honest node outputs a value $v$, then $v$ is valid w.r.t. $Q$, i.e., $Q(v) = 1$.
\end{lemma}
\begin{proof}
	This property is fairly trivial, as if it is violated, there must be $f+1$ honest nodes send \emph{READY(h)} with $\hash(v)=h$, further indicating that there must be $f+1$ honest nodes send \emph{ECHO(h)} with $\hash(v)=h$, which causes contradiction because honest nodes wouldn't send \emph{ECHO(h)} with $\hash(v)=h$ if $Q(v) \ne 1$.
\end{proof}

\noindent
{\bf Security of FIN-$\mvba$-Q}. We now prove that FIN-$\mvba$-Q realizes all four desired properties of $\mvba$, i.e., termination, agreement, external validity, and 1/2-quality.

\begin{lemma}\label{lemma:rbcdeliver}
	If all honest nodes input some values satisfying $Q$, there exists $n-f$ avw$\rbc$ instances, such that all honest nodes will eventually $wr\text{-}deliver$ in them.
\end{lemma}
\begin{proof}
	We consider two opposing scenarios: 1$)$ No honest node invoke $abandon()$ before all honest nodes $wr\text{-}deliver$ in the same $n-f$ avw$\rbc$ instances. 2$)$ Some honest node invoke $abandon()$ before all honest nodes $wr\text{-}deliver$ in the same $n-f$ avw$\rbc$ instances. For the former, \emph{Validity} property ensures all honest nodes eventually $wr\text{-}deliver$ in $n-f$ avw$\rbc$ instances leading by honest nodes. For the latter, if some honest node $\node_i$ invoked $abandon()$, the set $F_i$ of $\node_i$ must contain $n-f$ \emph{true} values. So at least $n-f$ avw$\rbc$ instances have been $wr\text{-}delivered$ by at least $f+1$ honest nodes. According to the \emph{Totality} property of avw$\rbc$, all honest nodes will eventually $wr\text{-}deliver$ in these $n-f$ avw$\rbc$ instances. 
\end{proof}

\begin{lemma}\label{lemma:rbcfinish}
	If all honest nodes input some values satisfying $Q$, all honest nodes will eventually step into iteration phase (line 13-28).
\end{lemma}
\begin{proof}
	According to Lemma~$\ref{lemma:rbcdeliver}$, there exists $n-f$ avw$\rbc$ instances (denoted by ${HR}_{n-f}$), such that all honest nodes will eventually $wr\text{-}deliver$ in them. So for every avw$\rbc$ instance in ${HR}_{n-f}$, at least $n-f$ nodes will multicast \emph{FIN} message responding to it. So all honest nodes will eventually receive $n-f$ \emph{FIN} messages from distinct nodes for every avw$\rbc$ instance in ${HR}_{n-f}$. Thus, all honest nodes will step into iteration phase.
\end{proof}

\begin{lemma}\label{lemma:quality}
	Using $H$ to denote the first $f$+$1$ honest nodes invoking $abandon()$. For any value $v$, its hash value $\hash(v)$ won't be $wr\text{-}delivered$ by any honest node, if responding \emph{VAL(v)} messages havn't been received by any of $H$ before $H$ all invoked $abandon()$.
\end{lemma}
\begin{proof}
	Since the reconstruction threshold of $Election()$ we use is $n-f$, at least $f+1$ honest nodes will have invoked $abandon()$ before adversary can learn the output of $Election()$. As shown in Figure \ref{algo:rbc}, honest node will reject sending \emph{ECHO} message for new received \emph{VAL} message after invoking $abandon()$. Thus for any value $v$ whose responding \emph{VAL} messages havn't been received by any of $H$ before they invoked $abandon()$, it can't gather enough (i.e. $n-f$) \emph{ECHO($\hash(v)$)} messages. Without enough \emph{ECHO($\hash(v)$)} messages, honest nodes won't send responding \emph{READY($\hash(v)$)} messages. So no honest node can receive sufficient ($n-f$)  \emph{READY($\hash(v)$)} messages to $wr\text{-}deliver$ $\hash(v)$.
\end{proof}

\begin{lemma}\label{lemma:agreement}
	\textbf{Agreement} If two honest nodes output $v$ and $v'$, respectively, then $v=v'$.
\end{lemma}
\begin{proof}
	According to the \emph{Agreement} property of $RABA$, all honest nodes will have same output for every possible running $RABA$ instances. So honest nodes will end iteration in the same repeat round (denoted by $R_{halt}$). Let $k$ be the result of $Election()$ of round $R_{halt}$. $RABA$ output 1 in round $R_{halt}$, so that at least one honest node $raba$-$propose$ or $raba$-$repropose$ 1. So that at least one honest node has $wr\text{-}delivered$ in $avw\rbc_k$. According to Theorem\ref{lemma:totality}, all honest nodes will $wr\text{-}deliver$ the same hash \emph{h}. According to the collision resistance property of hash function, there can't be a different value $v'$ that $\hash(v)=\hash(v')=h$.
\end{proof}

\begin{lemma}\label{lemma:exvalidity}
	\textbf{External-Validity.} If an honest node outputs a value $v$, then $v$ is valid w.r.t. $Q$, i.e., $Q(v) = 1$;	
\end{lemma}
\begin{proof}
	As shown in Lemma~\ref{lemma:agreement}, if an honest node outputs a value $v$, at least one honest node has $wr\text{-}delivered$ $\hash(v)$. So at least $f+1$ honest nodes have received message \emph{VALUE(v)} in responding avw$\rbc$ instance and sent \emph{ECHO($\hash(v)$)}. Since the condition of sending \emph{ECHO($\hash(v)$)} is $Q(v) = 1$, the \emph{External-Validity} property holds.
\end{proof}

\begin{lemma}\label{lemma:12quality}
	\textbf{1/2-Quality.} If an honest node outputs $v$, the probability that $v$ was input by the adversary is at most 1/2.
\end{lemma}
\begin{proof}
	Let node $\node_i$ be the first honest node who step into iteration phase. So at least $n-f$ values in set $F_i$ are true. Let $S$ denotes the set of avw$\rbc$ instances whose responding values in $F_i$ are true. $H$ denotes the set of avw$\rbc$ instances in $S$ who is leading by honest nodes. $A$ denotes the set of avw$\rbc$ instances in $S$ who is leading by adversary. $R$ denotes the set of avw$\rbc$ instances outside $S$. We now discuss the three possible cases of all possible iterations. Let $k_r$ donotes the result of $Election()$ of the $r_{th}$ iteration.
	
	\begin{enumerate}
		\item  $k_r$ falls in $A$. In this case, at least $f+1$ honest nodes will input 1 for $RABA_r$. According to the $\emph{Totality}$ of avw$\rbc$, all honest nodes eventually $wr\text{-}deliver$ in $avw\rbc_{k_r}$. And Lemma~\ref{lemma:rbcfinish} shows that all honest nodes will step into iteration phase. So the $\emph{Biased termination}$ property is met for $RABA_r$. The $\emph{Biased validity}$ property ensures all honest nodes output 1 for $RABA_r$. Then FIN-MVBA will eventually output the input value of $avw\rbc_{k_r}$, which is chosen by the adversary. Note that the probability of this scenario occurring is at most 1/3.
		\item  $k_r$ falls in $H$. Similar to the former case, FIN-MVBA will output the input value of $avw\rbc_{k_r}$, which is an input chosen by the honest node. Note that the probability of this scenario occurring is at least 1/3.
		\item  $k_r$ falls in $R$. In this case, Lemma \ref{lemma:quality} shows that no $new$ value can be delivered  after $f+1$ honest nodes invoking $abandon()$. So the attack we describe in Appendix \ref{app:quality-attack} won't work here. Since there is no guarantee on how many honest nodes $wr\text{-}deliver$ in $avw\rbc_{k_r}$. $RABA_r$ can output 1 or 0 (we will show why $RABA$ eventually outputs in this scenario in Lemma \ref{lemma:termination}). Outputing 1 won't help even if $avw\rbc_{k_r}$ is leading by adversary, as it will reduce the probability of case(1). So we only care about the scenarion where $RABA_r$ outputs 0, whose probability of occurrence is at most 1/3. And honest nodes will jump to next iteration in this scenario.
	\end{enumerate} 
Thus, during every possible iteration, $RABA$ outputs 1 with a probability of up to 2/3. And once $RABA$ outputs 1, the probability that the output of FIN-$\mvba$-Q belong to the inputs of adversary is up to 1/2.
\end{proof}

\begin{lemma}\label{lemma:termination}
	\textbf{Termination.} If all honest nodes input some values satisfying $Q$, then each honest node would output.
\end{lemma}
\begin{proof}
	According to Lemma~\ref{lemma:rbcfinish}, all honest nodes will eventually step into iteration phase. So we only care about the $\emph{Termination}$ property of $RABA_r$ in $r_{th}$ iteration. Let $k_r$ donote the result of $Election()$ of the $r_{th}$ iteration. As shown in Lemma \ref{lemma:12quality}, $RABA_r$ has at least 2/3 probability to meet $\emph{Biased termination}$ property of $RABA$, and thus terminates. We divide the remaining cases into two categories (case(3) in Lemma \ref{lemma:12quality}). 1$)$ All honest nodes $raba$-$propose$ 0 and never $raba$-$repropose$ 1 for $RABA_r$. Thus $RABA_r$ terminates due to $\emph{Unanimous termination}$ property. 2$)$ Some honest nodes (less than $f+1$) have $raba$-$propose$ or  $raba$-$repropose$ 1 for $RABA_r$. Thus at least one honest node has $wr\text{-}delivered$ in $avw\rbc_{k_r}$. According to the $\emph{Totality}$ property of avw$\rbc$, all honest nodes will eventually $wr\text{-}delivered$ in $avw\rbc_{k_r}$, and thus $raba$-$propose$ or  $raba$-$repropose$ 1 for $RABA_r$. $RABA_r$ then terminates due to the $\emph{Biased termination}$ property.
\end{proof}

\section{Proof of $\finng$}\label{app:finng-proof}
{\bf Security of $\finng$}. We now prove that $\finng$ realizes agreement, total order, and liveness, not only for these always-honest nodes but also for obliviously recovered nodes (that was corrupted).

\smallskip
We first prove $\finng$ is secure regarding those always-honest nodes (note that there are at least $n-f$ nodes that are always honest during the course of execution).
Later we call these nodes always-honest nodes or simply honest nodes.

\begin{lemma}\label{lemma:pulling}
	If an honest node $\node_i$ output $\height_{e+1}$, it can eventually $wr\text{-}delivered$ in all w$\rbc$ instances between $\height_e$ and $\height_{e+1}$. Let set ${Digests}_{e+1}$ denote all the digests $wr\text{-}delivered$ in those w$\rbc$ instances, and let
	set ${TXs}_{e+1}$ denote the original data responding to ${Digests}_{e+1}$. Then $\node_i$ can eventually obtain the whole ${TXs}_{e+1}$.
\end{lemma}
\begin{proof}
	According to Lemma~\ref{lemma:exvalidity}, if $\node_i$ output $\height_{e+1}$ from ${FIN\text{-}\mvba\text{-}Q}_{e+1}$, then $Q(\height_{e+1}) = 1$ for at least $f+1$ honest nodes. Such that at least $f+1$ honest nodes have $wr\text{-}delivered$ in all w$\rbc$ instances between $\height_e$ and $\height_e+1$. Due to the $\emph{Totality}$ property of w$\rbc$, all honest nodes will eventually $wr\text{-}delivered$ in all w$\rbc$ instances between $\height_e$ and $\height_e+1$. Besides, for every $h$ $wr\text{-}delivered$ in those w$\rbc$ instances, at least $f+1$ honest nodes have received $\emph{VAL(v)}$ such that $\hash(v)=h$. Thus $\node_i$ can use the pulling mechanism described in Section\ref{sec:jumbo} to obtain missing data.
\end{proof}

\begin{lemma}\label{lemma:fin-ng-honest}
	$\finng$ satisfies agreement, total order, and liveness  properties (regarding these always honest nodes) except with negligible probability.
\end{lemma}
\begin{proof}
	Here we prove the three properties one by one (for  those always-honest nodes).

\textbf{For agreement:} Suppose that one honest node $\node_i$ output a transaction $tx$. $\node_i$ must have finished two FIN-$\mvba$-Q instances with output $\height_e$ and $\height_{e+1}$, such that $tx$ belongs to some w$\rbc$ instance between $\height_e$ and $\height_{e+1}$. Due to the $\emph{Agreement}$ property of FIN-$\mvba$-Q, if another honest $\node_j$ output in the same two FIN-$\mvba$-Q instances, they output the same value. And $\node_j$ thus can obtain some $tx'$ responding to the same w$\rbc$ instance due to Lemma~\ref{lemma:pulling}. And the collision resistance property of hash function together with the $\emph{Agreement}$ property of w$\rbc$ ensure that $tx = tx'$. Beside, we also need to prove that if node $\node_i$ outputs in some FIN-$\mvba$-Q instance, all honest nodes will eventually active this instance. To prove it, we need the $\emph{Liveness}$ property of $\finng$ (we will prove it later). In the begining, all honest nodes will active the first FIN-$\mvba$-Q instance and terminate then due to the $\emph{Termination}$ property of FIN-$\mvba$-Q. An the $\emph{Liveness}$ property of $\finng$ ensures all honest nodes will eventually generate a new input satisfying $Q$, and then active the next FIN-$\mvba$-Q instance. So that all honest nodes will continuously active FIN-$\mvba$-Q instances one by one. Such that if $\node_i$ output a transaction $tx$, all honest nodes eventually output $tx$. 

\textbf{For Total-order:} In $\finng$, all honest nodes sequential participate in MVBA epoch by epoch, and in each MVBA, all honest nodes output the same block, so the total-order is trivially hold.

\textbf{For Liveness:} One honest node $\node_i$ can start $w\rbc_{e+1}$ once it $wr\text{-}delivered$ in $w\rbc_{e}$. The $\emph{Totality}$ property ensures that all honest nodes eventuall $wr\text{-}deliver$ in $w\rbc_{e}$. All honest nodes thus can active $w\rbc_{e+1}$. So that $\node_i$ won't be stuck. It also
implies at least $n-f$ parallel broadcasts can grow continuously since all honest nodes won't be stuck and can continuously active new w$\rbc$ instance. Such that $\node_i$ can eventually generate a new input to active a new FIN-$\mvba$-Q instance.

According to the \emph{1/2-Quality} of FIN-$\mvba$-Q, the input of honest nodes can be outputted with a probability not less than 1/2, so even in the worst case, once a FIN-$\mvba$-Q instance decides to output the input of an honest node $\node_i$, the broadcasted transactions proposed by $\node_i$'s FIN-$\mvba$-Q input  will be decided as the final consensus output. 
It further indicates that the adversary cannot censor the input payloads that are broadcasted by honest nodes, because the probability of FIN-$\mvba$-Q outputting the input of some honest node is at least 1/2, which implies that 
after $k$ executions of FIN-$\mvba$-Q, 
the probability that a finished broadcast is not yet output would be at most $(1/2)^k$, i.e., becomes negligible after sufficient repetition of FIN-$\mvba$-Q.

\end{proof}


\section{Proof of $\jumbo$}\label{app:jumbo}

Since the proofs of agreement and total order properties for $\jumbo$ resembles those of  $\dumbong$, we only describe the proof for liveness and fairness here.

 W.l.o.g, we assume that nodes disseminate the same number of transactions in each CBC instance. To enhance the fairness of $\jumbo$, we adopt a ``speed limit'' patch described  in Section \ref{sec:jumbo-impl}. For a better description, we do some changes on variables.  We use $\delta_{ij}$ to denote how many transactions output by node $\node_i$ in $j$-th broadcast thread since the height included by last block (from the local view of node $\node_i$). Let $\delta_i$ to track the $(f+1)$-th smallest one of all $\{\delta_{ij}\}_{j\in [n]}$. The variable $\delta_{low}$ tracks the smallest on of all $\{\delta_i\}_{i\in [n]}$.
 
 \begin{lemma}\label{lemma:newinput}
 	Any MVBA instance won't be blocked forever. In other words, all honest nodes will constantly output in MVBA thread.  
 \end{lemma}
 \begin{proof}
 	We proof this by mathematical induction:
 	\begin{enumerate}
 		\item W.l.o.g, we assume epoch number start from 0. Now we proof that all honest nodes will output in epoch 0. There are two possible cases:
 		\begin{enumerate}
 			\item All honest nodes stay in epoch 0 (they all output 0 block). We recall the condition of halt voting. $\node_i$ halts voting for the $j$-th broadcast thread, if $ \beta \cdot \delta_{ij} \ge  \delta_i$ and $\delta_{ij} > 0$. So no honest nodes will halt voting in the first CBC instance of each honest broadcast thread, as $\delta_{ij} = 0$ in their views before running that instance. Hence honest nodes will at least finish the first CBC instance of each honest broadcast thread, and thus meet the condition of generating a new input for a new MVBA instance. As a result, all 
 			honest nodes will eventually active the first MVBA instance. Thus they will all output in epoch 0 because of the \emph{Termination} property of MVBA.
 			 
 			\item Some honest nodes stay in epoch 0, while the rest stay in higher epochs. We use $N_{0}$ to denote the set of honest nodes who stay in epoch 0, and $N_{high}$ to denotes the rest honest nodes. Now we proof that all nodes in $N_{0}$ can finally generate new inputs for the first MVBA instance.
 			
 			We assume that $\node_i \in N_{high}$ and the $j$-th broadcast thread is honest. If $\node_i$'s outputs contain a higher height than 0 for the $j$-th broadcast thread. It must have voted in the first CBC instance of the $j$-th broadcast thread, and thus won't block other honest nodes there. If $\node_i$'s outputs don't contain a higher height than 0 for the $j$-th broadcast thread. Things are going similarly to the former case. As a result, all honest nodes in $N_0$ can at least finish the first CBC instance of each honest broadcast thread, and thus meet the condition of generating new a input for a new MVBA instance. Thus they will all output in epoch 0 because of the \emph{Termination} property of MVBA.
 		\end{enumerate}
 		
 		\item W.l.o.g, we assume all honest nodes have output in epoch $e-1$. Now we proof that all honest nodes will eventually output in epoch $e$. There are two possible cases:
 		\begin{enumerate}
 			\item All honest nodes stay in epoch $e$ (they output the same number of blocks). We recall the condition of halt voting. $\node_i$ halts voting for the $j$-th broadcast thread, if $ \beta \cdot \delta_{ij} \ge  \delta_i$ and $\delta_{ij} > 0$. So no honest nodes will halt voting in the first CBC instance of each honest broadcast thread since last block, as $\delta_{ij} = 0$ in their views before running that instance. Hence honest nodes will at least finish the first CBC instance of each honest broadcast thread, and thus meet the condition of generating a new input for a new MVBA instance. As a result, all honest nodes will eventually active the MVBA instance in epoch $e$. Thus they will all output in epoch $e$ because of the \emph{Termination} property of MVBA.
 			
 			\item Some honest nodes stay in epoch $e$, while the rest stay in higher epochs. We use $N_{low}$ to denote the set of honest nodes who stay in epoch $e$, and $N_{high}$ to denotes the rest honest nodes. Now we proof that all nodes in $N_{low}$ can finally generate new inputs for the MVBA instance of epoch $e$.
 			
 			We use $h_{ej}$ to denote the height of the first CBC instance of the $j$-th broadcast thread since the block committed in epoch $e-1$. We assume that $\node_i \in N_{high}$ and the $j$-th broadcast thread is honest. If $\node_i$'s outputs contain a higher height than $h_{ej}$ for the $j$-th broadcast thread. It must have voted in the $h_{ej}$-th CBC instance, and thus won't block other honest nodes who havn't output there. If $\node_i$'s outputs don't contain a higher height than $h_{ej}$ for the $j$-th broadcast thread. Things are going similarly to the former case. As a result, all honest nodes in $N_{low}$ can at least finish the first CBC instance of each honest broadcast thread, and thus meet the condition of generating new a input for a new MVBA instance. Thus they will all output in epoch $e$ because of the \emph{Termination} property of MVBA.
 		\end{enumerate}
 
 	\end{enumerate}
 	Thus all honest nodes will constantly output in MVBA thread.
 \end{proof}

\begin{lemma}\label{lemma:BCstuck}
All honest broadcast threads won't be blocked forever.
\end{lemma}
\begin{proof}
According the proof of \ref{lemma:newinput}, at least the first CBC instance of honest broadcast threads in each epoch won't be blocked. We assume that node $\node_i$ is honest, and $\node_i$ is in epoch $e$. Let $B_{1st}$ be the transaction block $\node_i$ propose in his first CBC instance after the block committed in epoch $e-1$. Then $B_{1st}$ is guaranteed to be output by all honest nodes according to the $Validity$ property of CBC. Because all honest nodes will constantly output in MVBA thread, all honest nodes will add $B_{1st}$ into their inputs. According to the 1/2-Quality of MVBA, $B_{1st}$ will be output after expected constant number of MVBA instances. Thus $\node_i$ will eventualy step into an epoch where $B_{1st}$ has been output before. $\node_i$ can then start a new CBC instance, which is guaranteed to be finished by all honesy nodes according to the previous description. Thus $\node_i$ can constantly active new CBC instances. So the $i$-th broadcast thread won't be blocked forever, and so does the rest honest broadcast threads.
\end{proof}

	W.l.o.g, we assume that a transcation $tx$ is placed in honest node $\node_i$'s input buffer. According to Lemma-\ref{lemma:BCstuck}, $\node_i$ will eventually propose it in one of his CBC instance. Thus $tx$ will be output by all honest nodes in that CBC instance. According the proof of \ref{lemma:newinput}, $tx$ will finally be contained in all MVBA input from honest nodes. After that, $tx$ will be output by consensus in expected constant number of MVBA instances. Thus the $Liveness$ property for $\jumbo$ is held.

\section{Choice of Experimental Environment}
We adopt a popular appraoch of using Linux TC to emulate WAN networks among EC2 instances in the same region \cite{yang2021dispersedledger,neiheiser2021kauri}. 
Namley, we apply manually-controlled delay and bandwidth throttling at each node. In contrast, a few recent studies \cite{miller2016honey,beat,guo2020dumbo,guo2022speeding,gao2022dumbo,tusk,duan2023practical,zhang2022pace} adopted a different way to set up benchmarking by evenly distributing EC2 instances at different AWS regions. The reasons behind our benchmarking choice are as follows: 
\begin{enumerate}

	\item {\em Economic affordability}. The price of data transfer between different regions for Amazon's EC2 platform is non-negligible, which is around $\$$0.02 per GigaByte. So if we evaluate protocol at the scale of 256 nodes and 1 Gbps bandwidth, the cost on data transfer is around $\$$38.4 per minute, as opposed to $\$$1.5 for computation. If we deploy all nodes at the same region, we can conduct the same experiments for 25 times longer at the same cost, because data transfer within the same region is   free. 
	Otherwise, if we distribute 256 nodes across a dozen of regions, the cost would become completely unafforable.
	
	\item {\em Result  reproducibility}. The network bandwidth and latency between different regions vary greatly over time (due to changes of AWS' network infrastructure), making it difficult for us to maintain the same network conditions when  repeating the same experiments.
	
	\item {\em Better emulation of WANs}. If we deploy several hundred of nodes in a dozen of different regions, then many nodes will inevitably be deployed in the same region. The network conditions between nodes within the same region are extremely good, typically with only a few milliseconds of latency and 10 Gbps bandwidth, which is completely mismatched with the WAN environment. If we deploy all nodes to the same region, we can use tools like TC to precisely control the bandwidth and latency of each connection, in order to better simulate the WAN environment. Additionally, we can more accurately simulate network fluctuations and other asynchronous adversarial attack behaviors.
	
	\item {\em Precise control of bandwidth}. The last benefit of our benchmarking choice is that we can precisely control the bandwidth  during each test. 
	This allows us compare our throughput to line-rate speed to understand the actual bandwidth utilization of our different protocols.

\end{enumerate}

\end{document}